\documentclass[11pt]{article}
\usepackage[margin=1in, centering]{geometry}
\usepackage{graphicx} 
\usepackage{xlmacros}
\usepackage{amsmath}
\usepackage{esint}
\usepackage{mathtools}
\usepackage{amssymb}
\usepackage{amsthm}
\usepackage{xcolor}
\usepackage{bbm}
\usepackage{enumitem}

\title{Cell-Probe Lower Bounds via Semi-Random CSP Refutation: Simplified and the Odd-Locality Case}
\author{Venkatesan Guruswami\thanks{Simons Institute for the Theory of Computing and Departments of EECS \& Mathematics, UC Berkeley. {\tt venkatg@berkeley.edu}. Research supported in part by a Simons Investigator award and NSF grant CCF-2211972.} \and Xin Lyu \thanks{UC Berkeley. {\tt xinlyu@berkeley.edu}. Supported by a Google PhD fellowship.} \and Weiqiang Yuan \thanks{EPFL. {\tt weiqiang.yuan@epfl.ch}. Supported by the Swiss State Secretariat for Education, Research and Innovation (SERI) under contract number MB22.00026.}}
\date{July 2025}

\begin{document}

\maketitle
\thispagestyle{empty}

\begin{abstract}
A recent work (Korten, Pitassi, and Impagliazzo, FOCS 2025) established an insightful connection between static data structure lower bounds, range avoidance of $\text{NC}^0$ circuits, and the refutation of pseudorandom CSP instances, leading to improvements to some longstanding lower bounds in the cell-probe/bit-probe models. Here, we improve these lower bounds in certain cases via a more streamlined reduction to XOR refutation, coupled with handling the odd-arity case. Our result can be viewed as a complete derandomization of the state-of-the-art semi-random $k$-XOR refutation analysis (Guruswami, Kothari and Manohar, STOC 2022, Hsieh, Kothari and Mohanty, SODA 2023), which complements the derandomization of the even-arity case obtained by Korten et al.

As our main technical statement, we show that for any multi-output constant-depth circuit that substantially stretches its input, its output is very likely far from strings sampled from distributions with sufficient independence, and further this can be efficiently certified. Via suitable shifts in perspectives, this gives applications to cell-probe lower bounds and range avoidance algorithms for $\NC^0$ circuits.

\end{abstract}

\section{Introduction}

A recent work by Korten, Pitassi, and Impagliazzo \cite{KPI25} improved some longstanding static data structure lower bounds, based on an insightful connection to range avoidance of $\mathsf{NC}^0$ circuits and refutation of pseudorandom CSP instances~\cite{KPI25}. In this work, we improve their cell-probe lower bounds in certain cases, using methods to refute odd-arity semi-random XOR, while streamlining some aspects of the proof, such as the reduction to $k$-XOR instances. Furthermore, our use of Fourier analysis for the reduction not only cleans up the bound, but also gives clear and direct improvements over \cite{KPI25}'s lower bounds (for \emph{bit}-probe data structures) and algorithms (for range avoidance) for \emph{all} ranges of arities (both odd and even).

We first state our main technical result, which says that for every simple circuit that stretches its input substantially, its output will almost certainly be far from a string sampled from any distribution with sufficient independence. Furthermore, this property can be certified (with high probability) via a polynomial-time algorithm. 

This refutation result then yields, via suitable shifts in perspective (including the recent connections made in \cite{KPI25}), improved data structure lower bounds and range avoidance algorithms for simple circuits.  We will only describe the context of these applications briefly, referring the reader to the excellent introduction of \cite{KPI25} for further discussion concerning these connections. 

\medskip\noindent\textbf{Roadmap of Introduction.} We introduce our main refutation result in \Cref{sec: intro certify remote}. Applications to data structure lower bounds and range avoidance algorithms will be delivered in \Cref{sec: intro data structure} and \Cref{sec: intro range avoidance}, respectively. Finally, \Cref{sec: intro technique} gives a broad overview of our proofs.



\subsection{Certifying Remote Points for Simple Circuits.} \label{sec: intro certify remote}

Given a circuit $C:\{\pm 1\}^n\to \{\pm 1\}^{m}$ and a string $y\in \{\pm 1\}^m$, the statement that $y\in \mathrm{Range}(C)$ is clearly in $\NP$. We are interested in the complexity of the negated statement ``$y\notin \mathrm{Range}(C)$.'' Observe that this is a $\mathbf{coNP}$ statement by definition. We ask: under which conditions can we certify that $y\notin \mathrm{Range}(C)$ with a computation less powerful than $\mathbf{coNP}$? Note that an analogous question has been asked and thoroughly investigated for CSP refutation: namely, given a SAT instance $\psi$, the claim that $\psi$ is satisfiable is clearly in $\NP$ and easily verifiable (given a witness), while $\psi$ being unsatisfiable is generally harder to verify. 
 
\paragraph*{Notation.} Before formally stating our main result alluded to above, we need some pieces of notation. For two Boolean strings $x,y\in \{\pm 1\}^m$ of the same length $m$, we use $\Delta(x,y) := \frac{1}{m} \sum_{i=1}^{m}\mathbf{1}\{x_i \ne y_i\}$ to denote their (relative) Hamming distance. 
A distribution $\mathcal{D}$ on $\{0,1\}^m$ is said to be $\eta$-almost $k$-wise independent distribution if the projection of $\mathcal{D}$ onto any $k$ coordinates is at most $\eta$-far (in infinity norm) from the uniform distribution on $\{\pm 1\}^k$ in total variation distance. 

We now state our main result.


\begin{theorem}[Main]\label{thm: main odd arity}
    There is a universal constant $c_{\text{remote}}>0$ such that the following is true. Let $k,t,w,n,m\ge 1$ be integer parameters satisfying $k\ge t\log n$ and let $\eps \in (0, 1)$. Further, let $\Sigma$ be a finite alphabet of size $2^w$, and $\mathcal{D}$ be a distribution over $\{\pm 1\}^{m}$ that is $\eta$-almost $k$-wise independent where $\eta \le (2^{-tw}\cdot \eps^{4}\cdot n^{-\frac{k}{\log n}})^{O(1)}$. 

    Let $C:\Sigma^{n}\to \{\pm 1\}^m$ be an arbitrary multi-output circuit. Suppose each output of $C(x)_i$, $1\le i\le m$, can be computed by a $t$-query (adaptive) decision tree over the input $x$. Then, provided that 
    $$m \ge c_{\text{remote}}\cdot n\cdot \biggl(\frac{n \log n}{k}\biggr)^{t/2-1}\cdot \log(n) \cdot \eps^{-4} \cdot 2^{O(tw)} \ , $$
    the following holds:  
    \begin{align*}
        \Pr_{\bm{b}\sim \mathcal{D}}\left[ \min_{x\in \Sigma^n} \{ \Delta(C(x, *), \bm{b}) \} \ge \frac{1}{2} - \eps \right] \ge 1 - \frac{1}{\mathrm{poly}(n)} \ .
    \end{align*}
    Furthermore, there is a deterministic algorithm running in time $\mathrm{poly}(m, n^{O(t)})$ such that, with probability $1-\frac{1}{\mathrm{poly}(n)}$ over $\bm{b}\sim \mathcal{D}$, the algorithm certifies\footnote{This is to say, the algorithm always outputs either ``certified'' or ``uncertain,'' and outputs  ``certified''  with probability $1-\frac{1}{\mathrm{poly}(n)}$ over the choice of $\bm{b}$. Whenever the algorithm outputs ``certified,'' the claimed bound on $\Delta(C(x, *), \bm{b})$ is guaranteed to hold.} the bound given the input $C$ and $\bm{b}$.
\end{theorem}

Near-optimal constructions of almost independent distributions are well-known \cite{NaorN93,AlonGHP92}. Specifically, the distribution $\mathcal{D}$ required in \Cref{thm: main odd arity} can be sampled with $O(\log(n) + k + \log(1/\eta)) \le O(k + \log(1/\eps) + tw + \log(n))$ bits and the sampling algorithm (given random seeds) is efficient. Consequently, \Cref{thm: main odd arity} implies that there exists an explicit ensemble of $\mathrm{poly}(n, 2^k, 1/\eps)$ Boolean strings of length $m$, 
such that for any circuit $C:\Sigma^n\to \{\pm 1\}^m$ obeying the form of \Cref{thm: main odd arity}, a randomly drawn string from the ensemble is far away from the range of $C$ with high probability.
In \Cref{sec: pseudorandomness}, we collect a list of pseudorandom distributions that we frequently use throughout the paper.


\paragraph*{Discussion.} We compare \Cref{thm: main odd arity} with what is known in the literature. When each output $C_i$ is a simple predicate (such as XOR, OR, NAE, etc.) over at most $t$ inputs, and $\bm{b}$ is drawn from the \emph{uniform} distribution, this is precisely the question of semi-random CSP refutation, and there has been a rather rich literature (see, e.g.~\cite{Feige07,AGK21,GKM,HsiehKM23}). In particular, it has already been observed in~\cite{Feige02,COCF10} that refuting general $t$-CSPs can be reduced to \emph{strongly}\footnote{``Strong refutation'' means to certify not only that the CSP instance is unsatisfiable, but also that any assignment cannot satisfy more than $\frac{1}{2} + \eps$ fraction of constraints.} refuting $t$-XOR instances. 

Motivated by the challenge of proving static data structure lower bounds with \emph{adaptive} query strategies, the recent work of \cite{KPI25} initiated the quest to investigate the case of $C_i$'s being low-depth decision trees and, more importantly, the case when $\bm{b}$ is drawn from a \emph{pseudorandom} distribution. \Cref{thm: main odd arity} represents a new advance in this direction. In particular, when $t$ is even, \Cref{thm: main odd arity} provides the same quantitative bound as the main result of \cite{KPI25}. The case of odd $t$'s was highlighted as an open problem in \cite[Open Problem 1]{KPI25}, which is affirmatively answered by \Cref{thm: main odd arity}. We thus bring all of what is known about semi-random CSP refutation with \emph{truly random} ``$\bm{b}$'' \cite{HsiehKM23} to the regime of \emph{pseudorandom} ``$\bm{b}$''. Furthermore, our proof of Theorem~\ref{thm: main odd arity}, while inspired by the result of \cite{KPI25}, offers a fresh perspective on the problem and provides a streamlined proof, as summarized below.

First, we use a principled Fourier analysis to reduce the question of refuting $C(x, *) = \bm{b}$ to strong refutation of $t$-XOR instances, in place of the combinatorial arguments employed by \cite{KPI25}. Our reduction provides a clear picture about what makes decision trees ``easy'' to refute, through the lens of their Fourier spectrum. In the second part of the proof, we make modular use of the state-of-the-art semi-random $t$-XOR refutation algorithm from \cite{HsiehKM23} (which in turn builds on \cite{GKM}):
We provide an exposition of \cite{HsiehKM23} and point out the key steps that allow for a pseudorandom ``$\bm{b}$'' to work. Overall, we believe our proof further clarifies and tightens the connection between semi-random CSP refutation and data structure lower bounds, an intriguing and beautiful connection first observed by \cite{KPI25}.



\subsection{New Data Structure Lower Bounds} \label{sec: intro data structure}

\paragraph*{The static cell-probe model.} In the static cell probe model~\cite{EF75, Yao81}, a data structure problem can be described by a mapping $f:[N]\times [m] \to \{0,1\}$. Think of $f$ as a matrix with $N$ rows and $m$ columns. The $N$ rows of $f$ describe the collection of strings required to be stored correctly (for example, for the task of efficient evaluation of low-degree $\mathbb{F}_2$-polynomials, each row is the evaluation table of a polynomial). A data structure solution for $f$ consists of an \emph{initialization} algorithm and a \emph{query} algorithm.
On initialization, the algorithm receives an index $i\in [N]$ (indexing the string of $f(i, *)$ to be stored), and produces a memory configuration $E\in \Sigma^S$ using $S$ cells, where each cell stores a word from a given alphabet $\Sigma$. In the query stage, the algorithm receives $j\in [m]$. Then it queries at most $t$ cells from $E$ and outputs the correct value of $f(i, j)$.

The efficiency of the data structure is measured by several parameters. The number of cells $S$ needed for storage is called the ``space complexity.'' The number of queries $t$ during the query stage is called the ``time complexity,'' and we call $\log |\Sigma|$ the ``word length.''
Lastly, depending on whether or not the query algorithm is allowed to make its $t$ queries adaptively, we have the distinction between adaptive and non-adaptive data structures. In the cell probe model, we do not assume computational upper bounds on the initialization and query algorithms, as we are primarily interested in the information-theoretic trade-offs between space $S$ and time $t$.

The counting argument~\cite{Mil93} demonstrates that most of the data structure problems $f:[N]\times m\to \ZO$ require either maximal space $S\approx m$, or maximal time $t\approx \log m$.
Much effort has then been devoted to construct an explicit data structure problem that satisfies these bounds.
However, the objective remains beyond our current reach:
There is a line of works~\cite{Siegel04,Mil93,PT06,Lar12,GGS23} dedicated to establishing increasingly stronger explicit data structure lower bounds.
However, the state-of-the-art lower bound is $S\ge \tilde{\Omega}(\log N\cdot (m/\log N)^{1/t})$ for adaptive cell-probe models~\cite{Siegel04} and $S\ge \tilde{\Omega}(\log N\cdot (m/\log N)^{1/{t-1}})$ for nonadaptive ones~\cite{GGS23} .

In this work, we focus on the parameter regime that $\log N \ll m$. 
In this setting, the best lower bound was $S\ge \tilde{\Omega}(m^{1/t})$ inherited from the aforementioned bound given by~\cite{Siegel04}, until very recently,
Korten, Pitassi, and Impagliazzo~\cite{KPI25} established a significantly improved bound of $S\ge \tilde{\Omega}(m^{2/t})$ for every even $t$.

\paragraph*{New lower bounds.} Having defined the model, we are ready to describe our new lower bounds. First, a reparameterization of Theorem~\ref{thm: main odd arity} and a perspective shift give us the following: Briefly, the circuit is given by the execution of the query-answering algorithm for the data structure. The circuit operates on an array of $S$ cells and thus $S$ plays the role of $n$ in \Cref{thm: main odd arity}.

%
\begin{theorem} \label{thm: data structure lb}
    Let $f:[N]\times [m]\to \{0,1\}$ be a data structure problem whose rows form the support of an $\eta$-almost $k$-wise independent distribution where $\eta \le (2^{-kw}\cdot m^{-\frac{k}{\log n}})^{O(1)}$. Let $S,t,w\in \mathbb{N}$ be parameters and suppose $k\ge t\cdot\log m$. For any space $S$, time $t$, word length $w$ cell-probe data structure algorithm for $f$, we must have
    \begin{align*}
        S \ge \frac{m^{\frac{2}{t}}\cdot k^{1-\frac{2}{t}}}{ 2^{O(w)}\cdot \log m}.
    \end{align*}
\end{theorem}

Theorem~\ref{thm: data structure lb} improves over \cite[Theorem 1]{KPI25}, which gave the bound but only for even $t$'s. We are able to cover the whole range of $t$'s.
We also remark that the bound in~\Cref{thm: data structure lb} is almost tight: Siegel~\cite{Siegel04} constructed a function $f:[N]\times [m]\to \ZO$, such that the rows of $f$ form the support of a $k$-wise independent distribution for any $k=o(n)$. Moreover, $f$ can be computed by a non-adaptive bit-probe data structure algorithm with time $t$ and space $S=m^{\frac{2}{t}+o(1)}$.

\paragraph*{The bit-probe lower bounds.} We now switch our focus to the bit-probe model. This simply means that we will work with $\Sigma = \{0,1\}$: i.e., each cell now stores a single bit.
We will work with small-biased distributions \textemdash distributions that fool any parity functions \textemdash which enables us to break through the $m^{2/t}$ barriers~\cite{Siegel04}. 

\begin{restatable}{theorem}{AdaptiveDSLB} \label{thm: adaptive-ds-lb}
Let $f:[N]\times [m]\to \ZO$ be a data structure problem whose rows form the support of a $2^{-\Omega(t)}n^{-t}$-biased distribution.
Then any adaptive bit probe data structure for $f$ with time $t$ requires space
\[
S\ge \tilde{\Omega}(m^{\frac{1}{\frac{t}{2}-\frac{t-2}{2(t+2)}}}),
\]
where $\tilde{\Omega}$ hides a $\polylog(m)$ factor and some constant that depends on $t$.
\end{restatable}

We can obtain a slightly better bound for non-adaptive bit probe data structures.

\begin{restatable}{theorem}{NonAdaptiveDSLB} \label{thm: nonadaptive-ds-lb}
Let $f:[N]\times [m]\to \ZO$ be a problem whose rows form the support of a $2^{-\Omega(t)}n^{-t}$-biased distribution.
Then any nonadaptive bit probe data structure for $f$ with time $t$ requires space
\[
S\ge \tilde{\Omega}(m^{\frac{2}{t-1}}),
\]
where $\tilde{\Omega}$ hides a $\polylog(m)$ factor and some constant that depends on $t$.
\end{restatable}
Prior to this work, the best known lower bound was given by Korten et al.~\cite{KPI25}. They show that for any $f$ as in~\Cref{thm: adaptive-ds-lb,thm: nonadaptive-ds-lb} and even $t$,
\begin{enumerate}[noitemsep]
\item any adaptive bit probe data structure for $f$ with time $t$ requires space $S\ge \tilde{\Omega}(n^{\frac{2}{t-1}})$, and
\item any nonadaptive bit probe data structure for $f$ with time $t$ requires space $S\ge \tilde{\Omega}(m^{\frac{1}{\frac{t}{2}-\frac{t-2}{2(t+2)}}})$.
\end{enumerate}
Compared to~\cite{KPI25}, our improvements are twofold: We provide tighter bounds for both nonadaptive and adaptive data structures, and our bounds apply to the full range of $t$'s.


\subsection{Better Algorithms for Range Avoidance} \label{sec: intro range avoidance}

In the range avoidance problem~\cite{KKMP21-TFNP-Hierarchy,Korten21}, one is given a circuit $C:\{0,1\}^{n}\to \{0,1\}^m$ with $n < m$ and tasked with identifying a string $y\in \{0,1\}^m$ outside the range of $C$. 
Recent years have seen a growing body of research \cite{Korten21,RSW22,GLW25,CHLR23,ILW23,GGNS23,CHR24,Li24,CL24,KP24,GC25,LZ25} dedicated to a systematic algorithmic and complexity-theoretic study of range avoidance, and its connections to other areas such as explicit constructions, circuit lower bounds, and total search problems in NP (TFNP).
We refer the reader to~\cite{Korten25} for an excellent overview.

Special attention has been devoted to the case where each output $C_i$ is computed by an $\NC_0^t$-circuit (e.g.,~\cite{RSW22, GLW25, GGNS23, KPI25,LZ25}). Recall that $\NC_0^t$ denotes circuits that compute their outputs from a subset of at most $t$ input bits. We use $\NC_0^t$-Avoid the denote this special class of range avoidance instances henceforth. In particular, when $t = 2$, $\NC_0^2$-Avoid turns out to be easy \cite{GLW25} so long as $m\ge n + 1$. When $t \ge  3$, it has been identified that the \emph{stretch} (i.e., how $m$ compares with $n$) of the circuit $C$ is closely tied to the hardness of the problem. \cite{GLW25} offered a $\mathbf{FP}^\NP$ algorithm to solve $\NC_0^t$-Avoid whenever $m \gg n^{t-1}$, this was subsequently improved in \cite{GGNS23}, where quantitatively a ``log'' factor was shaved and qualitatively the $\NP$ oracle removed. \cite{KPI25} gave an algorithm for $\NC_0^t$-Avoid whenever $m\gg n^{t/2}$ for even $t$. Equally important to algorithmic developments are complexity-theoretic barriers: it was shown in \cite{RSW22, GGNS23} that solving $\NC_0^t$-Avoid with sufficiently small stretch ($m \approx n+n^{o(1)}$ for $t = 4$ and $m \approx n + n^{2/3}$ for $t = 3$) implies breakthrough circuit lower bounds, a widely acknowledged challenge in circuit complexity.

\paragraph*{Range avoidance for circuit with smaller stretch.} We now describe our results for the range avoidance problem. We start with a polynomial-time algorithm.

\begin{restatable}{theorem}{PolyAlgoForAvoid}\label{thm: poly_algo_for_Avoid}
There is a universal constant $c_{\text{avoid}}$ such that the following holds.
This is a deterministic algorithm which, given $t\ge 3$ and  an $\NC^0_t$ circuit $C:\ZO^n\to \ZO^m$ as input, outputs some $y\in \ZO^m\setminus \range(C)$ in time $n^{O(t)}$ whenever $m\ge c_{\text{avoid}}^t\cdot n^{(t-1)/2}\log n$.
\end{restatable}

\Cref{thm: poly_algo_for_Avoid} improves the state-of-the-art $\NC_0^t$-Avoid algorithms for all $t\ge 3$ \cite{KPI25}. Before our result, the best deterministic algorithm can only handle stretch $m\gg n^{\left\lceil t/2\right\rceil}$, while the best $\mathbf{FP}^\NP$ algorithm can only handle $m\gg n^{\frac{t}{2} - \frac{t-2}{2(t+2)}}$ for odd $t$'s. Our algorithm improves on both results. We note that our Theorem~\ref{thm: poly_algo_for_Avoid} in turn answers \cite[Item (a) of Open Problem 3]{KPI25}.

Plug in $t = 3$ in \Cref{thm: poly_algo_for_Avoid}: we see that we now only need the stretch to be super-linear by a single log factor (i.e.~$m\gg n\log n$) to solve $\NC_0^3$-Avoid efficiently (in polynomial time, without an $\NP$-oracle). It was known in~\cite{GGNS23} that solving $\NC_0^3$-Avoid with stretch $n + O(n^{2/3})$ implies constructions of rigid matrices (and thus super-linear circuit lower bounds). Thus, for $t=3$, the ``algorithmic threshold'' has been brought close to the complexity-theoretic barrier. Unfortunately, the gap between algorithms and complexity-theoretic barriers remains large for $t\ge 4$. It is therefore a potentially fruitful direction to further understand algorithms and barriers for $\NC_0^t$ with $t\ge 4$, with the ultimate goal of determining the ``computational threshold'' for the problem. This might be akin to the determination of the ``spectral threshold'' of $n^{k/2}$ for semi-random $k$-XOR refutation: only above $m = n^{k/2}$ do we know of efficient algorithms for refutation. However, as we have seen in \Cref{thm: poly_algo_for_Avoid}, the threshold of $n^{t/2}$ is not a barrier for $\NC^t_0$-Avoid, and there are potentially more algorithmic surprises to be found.

\paragraph*{Sub-exponential algorithms and even smaller stretch.} In the sub-exponential time regime, a recent work of Li and Zhong~\cite{LZ25} exhibits an algorithm for $\NC^0_t$-Avoid with near-linear stretch.  Specifically, for any $\epsilon>0$, their algorithm runs in time $2^{n^{1-\frac{\epsilon}{t-1}+o(1)}}$ when stretch $m=n^{1+\epsilon}$. We provide a slightly more efficient algorithm based on our \Cref{thm: main odd arity}. 
\begin{restatable}{theorem}{SubEXPAlgoForAvoid}
There is an algorithm which, given $t\ge 4,\epsilon>0$, and an $\NC^0_t$ circuit $C:\ZO^n\to \ZO^m$ as input, outputs some $y\in \ZO^m\setminus \range(C)$ in time $2^{n^{1-\frac{2\epsilon}{t-3}+o(1)}}$ whenever $m\ge n^{1+\epsilon}$.
\end{restatable}

\subsection{Our techniques} \label{sec: intro technique}

In this subsection, we give an overview of our proof. We specialize to $\NC_0^t$ circuits (corresponding to non-adaptive data structures) $C:\{\pm 1\}^{n}\to \{\pm 1\}^m$. Note that here we also switch to the $\{\pm 1\}$ representation of bits, to facilitate Fourier analysis.

\subsubsection{Reduction to $t$-XOR Refutation}

Let $\bm{b}\in \{\pm 1\}^m$ be a randomly drawn string, claimed to be in the range of $C$. This means there is a string $x\in \{\pm 1\}^n$ such that $C(x) = \bm{b}$. Note that this already gives a CSP instance on variables $x$, where each output bit $i$ places a constraint that $C(x)_i = b_i$. Also recall that $C$ is a $\NC_0^t$ circuit, meaning that the constraints have locality $t$. Thus, the problem of refuting $\bm{b}\in \mathrm{range}(C)$ is exactly that of refuting a $t$-CSP instance induced from $C$ and $\bm{b}$.

\paragraph{The ``XOR principle.''} 
Even better, we can reduce the problem to strong refutation of $t$-XOR instances, following the so-called XOR principle (see, e.g.,~\cite{Feige07,GKM}). Toward that, write
\begin{align}
    \E_{\bm{i}\sim [m]}[ C(x)_{\bm{i}}\cdot b_{\bm{i}} ] = 1. \label{equ: intro correlation bound}
\end{align}
Use $C^i(x)$ to denote $C(x)_i$ for every $i \in [m]$. Given $C^i$, we may look at its Fourier expansion:
\begin{align*}
    C^i(x) = \sum_{\alpha\subseteq [n]} \widehat{C^i_\alpha} \cdot x_\alpha,
\end{align*}
where we use the convention that $x_{\alpha} = \prod_{i:\alpha_i = 1}x_i$.

Since each $C^i$ is a $t$-junta, it has at most $2^{t}$ non-zero Fourier characters, each of degree at most $t$. All of them are multiples of $2^{-t}$ and do not exceed $1$ in absolute value. We can therefore order the characters in an arbitrary list $C_{\alpha(i, 1)}^i x_{\alpha(i,1)},\dots,C_{\alpha(i, 2^t)}^i x_{\alpha(i, 2^t)}$. Use linearity of expectation to re-write \Cref{equ: intro correlation bound}:
\begin{align}
    1=\E_{\bm{i}\sim [m]}[ C(x)_{\bm{i}}\cdot b_{\bm{i}} ] = \sum_{\ell=1}^{2^t} \E_{\bm{i}\sim [m]}[ C^{\bm{i}}_{\alpha(\bm{i}, \ell)} x_{\alpha(\bm{i},\ell)} \cdot b_{\bm{i}} ]. \label{equ: intro correlation bound expanded}
\end{align}
It is clear that for each $\ell$, the inner expectation of \Cref{equ: intro correlation bound expanded} can be viewed as a (weighted) $t$-XOR instance. By the averaging principle, the claim of $C(x) \equiv \bm{b}$ implies that at least one of the $t$-XOR instances has value at least $2^{-t}$. Hence, to refute the claim $C(x)=\bm{b}$, it suffices to \emph{strongly} refute \emph{every} $t$-XOR instance produced by the reduction.

\paragraph*{Extensions to decision trees.} \Cref{equ: intro correlation bound expanded} is the starting point of all our reductions. We now mention two ways we extend the basic observation of \Cref{equ: intro correlation bound expanded} and how they lead to our improved results in various contexts.

We will begin with the case of decision trees. While the above argument behind \Cref{equ: intro correlation bound expanded} remains mostly valid if the circuits $C:\{\pm 1\}^n\to \{\pm 1\}^m$ 
in question are depth-$t$ decision trees operating on Boolean inputs, for applications to data structure lower bounds, it is also of interest to investigate decision trees querying non-Boolean inputs $x\in \Sigma^n$. There is a straightforward way of converting a $t$-query decision tree on input $x\in \Sigma^n$, to a $(tw)$-query decision tree on input $\widetilde{x}\in \{\pm 1\}^{tw}$ (recall that $w = \log|\Sigma|$ denotes the word length). After the conversion, one can run the reduction of \Cref{equ: intro correlation bound expanded} naively but the result would be a collection of $(tw)$-XOR instances. Nevertheless, there are additional structures in the conversion we can take advantage of: namely the resultant decision tree only has $t$ steps of adaptivity and in every step it can only query a (pre-determined) bundle of $w$ bits, instead of arbitrarily chosen $w$ bits. In \Cref{sec: reduction for DTs}, we show how to exploit these structural properties and perform a more efficient reduction, thereby building stepping stones toward Theorems~\ref{thm: main odd arity} and \ref{thm: data structure lb}.

\paragraph*{Improvement from handling degree-$t$ terms differently.} We now present the common idea behind the improvements of Theorems~\ref{thm: adaptive-ds-lb}, \ref{thm: nonadaptive-ds-lb} and \ref{thm: poly_algo_for_Avoid}. Recall that a general lesson we learned from the $t$-CSP refutation literature (e.g.~\cite{AGK21,GKM,HsiehKM23}) is that efficient refutation of semi-random $t$-XOR is possible whenever there are more than $m\gg n^{t/2}$ constraints. This quantitative relationship is also reflected in our \Cref{thm: main odd arity}. To break through the $n^{t/2}$ threshold, our idea starts by re-writing \Cref{equ: intro correlation bound expanded} as
\begin{align}
    1=\E_{\bm{i}\sim [m]}[ C(x)_{\bm{i}}\cdot b_{\bm{i}} ] = \E_{\bm{i}\sim [m]}\left[ \sum_{\alpha: |\alpha| = t} C^{\bm{i}}_{\alpha} x_{\alpha} \cdot b_{\bm{i}} \right] + \E_{i\sim [m]} \left[ \sum_{\alpha:|\alpha| < t} C_\alpha^{\bm{i}} x_\alpha\cdot b_{\bm{i}}.\right]. \label{equ: intro correlation bound high low decompose}
\end{align}
With only $m\approx n^{(t-1)/2}$ output bits, we can strongly refute the second term of \Cref{equ: intro correlation bound high low decompose} (or, equivalently, certify an upper bound of $\eps$ for the second term, where $\eps$ can be made arbitrarily small). Therefore, to refute \eqref{equ: intro correlation bound high low decompose}, it is sufficient to show a barely non-trivial (i.e., $1-\eps$) upper bound for the first term (under the condition that $m \ll n^{t/2}$, as otherwise the problem is easy).

In fact, we \emph{automatically} get a $1-2^{-t}$ upper bound on the expectation if each output of $C$ is a \emph{non-parity} $t$-junta (this corresponds to $\NC_0^t$ circuits and non-adaptive data structures): to see this, note that we have $\hat{C}_{\alpha}^i < 1-2^{-t}$ for every $i\in [m]$, if $\alpha\subseteq [n]$ is the $t$ coordinates that the non-parity function $C^i(x)$ depends upon. On the other hand, if there are more than $n$ output bits that are parity of the inputs, those outputs must lie in an affine subspace of co-dimension at least $1$. Consequently, they fail to support any low-biased distribution. These pieces of observation directly lead to our proofs of Theorems~\ref{thm: nonadaptive-ds-lb} and \ref{thm: poly_algo_for_Avoid}, solving \cite[Item (a) of Problem 3]{KPI25} along the way.

When each output bit of $C^i$ is computed by a depth-$t$ decision tree over \emph{Boolean} inputs, it is no longer true that $\sum_{\alpha:|\alpha| = t} C^i_{\alpha} < 1-\Omega(1)$. A simple example is given by the \emph{index} function: consider the function $f:\{\pm 1\}^{3}\to \{\pm 1\}$ defined as $f(x_1,x_2,x_3) = x_{\frac{x_1+5}{2}}$. This $f$ is computed by a decision tree of depth $2$ and one can verify that $\hat{f}_{\{x_1,x_2\}} + \hat{f}_{\{x_1, x_3\}} = 1$. There are many more such non-parity examples. Thus, even if there are more than $n$ output bits of this form, one cannot use a simple linear algebra argument to address the challenge.

Nevertheless, for $t$-query decision trees we still have the bound $\sum_{\alpha:|\alpha|=t} |C^i_{\alpha}| \le 1$. Pretend for the moment that each $i\in [m]$ comes with a single character $C^i_{\alpha}$ of degree $t$ and with weight $|C^i_{\alpha}| = 1$. In this case, it suffices to prove an upper bound of $1-\eps$ for standard $t$-XOR instances. This turns out to be a question already considered in \cite{KPI25}, and our solution will be similar to theirs. Both solutions take inspiration from related developments in the literature on semi-random $t$-XOR refutation and hypergraph Moore bounds \cite{HsiehKM23}. We refer the reader to the proof of \Cref{thm: adaptive-ds-lb} for more details. 

\subsubsection{Refuting Semi-Random $t$-XOR with Bounded Independence}

We have presented an overview of the reduction from general $t$-CSP refutations to $t$-XOR refutations. The remaining puzzle is to show that efficient refutation of semi-random $t$-XOR instances is possible even if the ``right-hand-side'' vector $\bm{b}$ is drawn from a distribution with limited independence. We will briefly describe the idea here and defer the formal treatment to \Cref{sec: refute}.

Our analysis is based on \cite{HsiehKM23}, who used a trace method to strongly refute $t$-XOR instances when $\bm{b}$ is drawn from the uniform distribution. Very roughly speaking, the analysis of \cite{HsiehKM23} boils down to understanding the spectral norm of a random matrix $A_{\bm{b}} = \sum_{i=1}^m \bm{b}_i A_i$ where $A_i$'s are fixed matrices induced from the ``left-hand-side'' of the $t$-XOR instance (the topology between variables and constraints), and $\bm{b}$ is the ``right hand side.'' It was shown in \cite{HsiehKM23} that $\|A_{\bm{b}}\|_{2\to 2}$ is small with high probability when $\bm{b}$ is uniformly chosen. This was proved via analyzing the quantity $\E_{\bm{b}}[ \mathrm{trace}( (A_{\bm{b}})^\ell ) ]$ for some $\ell \approx t\log(n)$. The analysis was done by \emph{treating the matrix trace in question as a degree-$\ell$ polynomial of $\bm{b}$}. 
We can therefore immediately deduce that the analysis of \cite{HsiehKM23} works equally well when we draw $\bm{b}$ from an $\ell$-wise independent distribution instead. Refer to \Cref{sec: refute} for a formal exposition of \cite{HsiehKM23}'s proof and how we modify their proof for our purpose.

\paragraph*{Remark: weighted vs.~unweighted $t$-XORs.} In \Cref{equ: intro correlation bound expanded}, observe that for each fixed $\ell$, the expectation can be understood as a weighted $t$-XOR instance where each constraint $x_{\alpha(i,\ell)} = b_i$ comes with a weight $C_{\alpha(i,\ell)}^i \in [-1, 1]$. This is slightly different from the standard $t$-XOR instances, where every constraint has a unit weight. Naturally, one wonders how it affects the analysis and the conclusion. Intuitively, having a bounded weight does not make the task of refutation harder: for example, when every weight $|C_{\alpha(i,\ell)}^i|$ is bounded away from $1-\Omega(1)$, provably (and without any sophisticated proof) the $t$-XOR instance has a value bounded away from $1$: this is the case even if all constraints can be simultaneously satisfied!

Indeed, looking into the proofs from \cite{GKM, HsiehKM23}, it is easy to see that the argument works equally well if every constraint comes with a bounded weight. If one insists on using \cite{GKM,HsiehKM23} as a black box in subsequent analyses, another way to get around the technicality is by the observation that each $C^i_{\alpha(i,\ell)}$ in \Cref{equ: intro correlation bound expanded} is a multiple of $2^{-t}$, allowing us to split the coefficients into multiple coefficients with weight $\pm 2^{-t}$ if needed. Finally, to handle possible negative weights, we can flip the corresponding bit of the right-hand-side and change every weight to a positive one. Since flipping each bit or not is independent of the realization of $\bm{b}$ (as the weights are solely determined by the circuits $C$), all pseudorandomness properties considered in this paper regarding the right-hand-side (e.g., $k$-wise independence, $\eps$-biased, etc) are preserved.
After this reduction, we obtain $t$-XOR instances with unit weights on each individual constraint but a $2^{-t}$ ``global'' weight on the whole instance. We then deploy the analysis for the unweighted case.




\section{The XOR Principle and Reductions} \label{sec: reduction for DTs}

To prove Theorem~\ref{thm: main odd arity}, we will reduce the question of analyzing $C$ to analyzing a related family of semi-random $t$-XOR instances. Let us begin with necessary pieces of notation.

\paragraph*{Notation.} 
A $k$-uniform hypergraph $\mathcal{H}$ over a vertex set $V$ is a collection of hyperedges over $\mathcal{H}$. Each hyperedge $C\in \mathcal{H}$ is simply a $k$-subset of $V$, denoted by $e_C$. We note that some prior works (implicitly) assume there are no parallel edges in $\mathcal{H}$, and they use $C$ to both index the hyperedge and refer to the subset of vertices $C$ connects. We introduce the notation of $e_C$ to handle possible multiple edges that arise in our reduction.

Having defined hypergrpahs, the definition of $k$-XOR instances is immediate.

\begin{definition}\label{def: k XOR instance}
        For any $k\ge 1$, a $k$-XOR instance $\psi_{\mathcal{H},b}$ consists of a $k$-uniform hypergraph $\mathcal{H}$ over the vertex set $[n]$ and a signing pattern $b\in \{\pm 1\}^\mathcal{H}$. We identify vertices of $\mathcal{H}$ with variables of $\psi$. For any assignment $x$, the value of $\psi$ on $x$ is $\psi_{\mathcal{H},b}(x) = \E_{\bm{C}\sim H}[ x_{e_{\bm{C}}}\cdot b_{\bm{C}}]$. We also denote
        \begin{align*}
            \mathrm{val}(\psi_{\mathcal{H},b}) = \max_{x\in \{\pm 1\}^n} \left\{ \lvert \psi_{\mathcal{H},b}(x) \rvert \right\}.
        \end{align*}
\end{definition}

We will frequently analyze a collection of $k$-XOR instances $\{\psi_{\mathcal{H},b}\}_{b}$ with the same topology $\mathcal{H}$ but with different signing patterns $b$ (this is known as the ``semi-random'' CSP in literature). We therefore define $k$-XOR scheme as a ``partial'' $k$-XOR instance with fixed $\mathcal{H}$ but indefinite $b$.

\begin{definition}\label{def: k XOR scheme}
    For any $k\ge 1$, a $k$-XOR scheme $\psi_{\mathcal{H},*}$ is a $k$-uniform hypergraph $\mathcal{H}$. Each vertex of $\mathcal{H}$ is a variable of $\psi$, and each hyperedge $C\in \mathcal{H}$ is a constraint among the $k$ variables in $e_C$ with undetermined signing. For every $b\in \{\pm 1\}^{\mathcal{H}}$, $\mathcal{H}$ and $b$ jointly determine a $k$-XOR instance $\psi_{\mathcal{H},b}$.
\end{definition}

Lastly, we remark that the above notations extend naturally to weighted $k$-XORs $\psi_{\cH,b}$, where $b\in [-1,1]^{\cH}$ now represents a real-valued vector.

\subsection{The Main Reduction}

With the notation and definitions in place, we are ready to state the main reduction presented in this section.

\begin{lemma}\label{lemma: XOR reduction}
    Let $w,t,n,m\ge 1$ be parameters and $\Sigma$ a finite alphabet of size $2^w$. Let $C:\Sigma^n\to \{\pm 1\}^m$ be a multi-output function, computed by a $t$-query adaptive decision tree over the input $x\in \Sigma^n$. Then, there is a circuit $\widetilde{C}: \{\pm 1\}^{nwt} \to \{\pm 1\}^m$, and an ensemble of $4^{tw}$ $t$-XOR schemes $\psi_{\mathcal{H}^1,*},\dots, \psi_{\mathcal{H}^{4^{tw},*}}$, all over the variable set $[ntw]$. The following claims hold for $\widetilde{C}$ and the ensemble.
    \begin{itemize}
        \item $\mathrm{Range}(\widetilde{C}) \supseteq \mathrm{Range}(C)$. 
        \item For every $b\in \{\pm 1\}^m$, if $d(\widetilde{C}(y, *), b) \ge \frac{1}{2} + \eps$ for some $y\in \{0, 1\}^{wnt}$, then there exists $\ell\in [4^{tw}]$ such that $\mathrm{val}(\psi_{\mathcal{H}^\ell, b}) \ge \frac{\eps}{4^{tw}}$.
    \end{itemize}
    Furthermore, $\widetilde{C}$ and $\psi_{\mathcal{H}^\ell,*}$ can be computed in time $2^{O(tw)}\cdot (n + m)$ from $C$.
\end{lemma}

The rest of the section proves Lemma~\ref{lemma: XOR reduction}. We will first show the construction of $\widetilde{C}$ and verify Item 1 of Lemma~\ref{lemma: XOR reduction}. Then we show how to obtain the ensemble of XOR schemes from $\widetilde{C}$ with the desired property of Item 2. 

\subsection{Layer-Respecting Circuits and Layered Inputs}

We now carry out the first step of the proof. Namely, let us first modify $C$ to have a nice, ``layered'' structure and have Boolean inputs. Identify $\{\pm 1\}^w$ with $\Sigma$ arbitrarily. Then, we replace the input space from $\Sigma^n$ to $\{\pm 1\}^{(w\cdot n)\cdot t}$ with the following understanding: the input space consists of $t$ layers $T^{(1)}, T^{(2)}, \dots, T^{(t)}$, each consisting of $n$ groups $T^{(i)}_{1},\dots, T^{(i)}_n$ where every group has $w$ bits. Given an input $\tilde{x}\in \{\pm 1\}^{(w\cdot n) t}$, we will index layers (resp.~groups) by $\tilde{x}_{T^{(i)}}$ (resp.~$\tilde{x}_{T^{(i)}_j}$). The restriction of $\tilde{x}$ on a group $\tilde{x}_{T^{(i)}_j}\in \{\pm 1\}^w$ is understood to represent a symbol from $\Sigma$.

Now we shall modify $C$ to get a circuit $\tilde{C}$ operating on inputs from $\{\pm 1\}^{wnt}$. For any input $x\in \Sigma^{n}$, there will be a corresponding input $\tilde{x}\in \{\pm 1\}^{wnt}$ to be constructed momentarily. We first explain how $\tilde{C}$ simulates $C$: Recall that every output of $C$ is computed by a $t$-query decision tree. Our simple modification is the following: whenever $C$ is to make its $t'$-th query to the $j$-th symbol of $x$ (i.e.,~$x_j$), we let $\tilde{C}$ make its $t'$-th query to the $j$-th group in the $t'$-th layer of $\tilde{x}$. Namely $\tilde{x}_{T^{(t')}_{j}}$. Since $\tilde{x}_{T^{(t')}_{j}}$ represents a symbol from $\Sigma$, $\tilde{C}$ can just simulate the behavior of $C$, to either make the next query or return its decision if $t' = t$.

It should be clear that given $x\in \Sigma^{n}$, we can construct a corresponding $\tilde{x}$ in the new input space: first we convert each symbol of $x$ to the corresponding string from $\{\pm 1\}^{w}$, which gives a string $\{\pm 1\}^{wn}$. We make $t$ duplicate copies of this $(nw)$-bit string and get a layered input $\tilde{x}\in \{\pm 1\}^{nw\cdot t}$ (consisting of $t$ identical layers, in fact). By definition, we see that $C(x)\equiv \tilde{C}(\tilde{x})$. Hence, the range of $\tilde{C}$ is no less than $C$. It remains to find the appropriate ensemble of $k$-XOR schemes for $\tilde{C}$. Circuits of the form $\tilde{C}$ are called \emph{layer-respecting} circuits henceforth.

\subsection{Fourier Analysis}

This subsection carries out the second step of the reduction.

\paragraph*{Notation.} For notational clarity, from now on we will drop the ``tilde'' from $\tilde{C}$, $\tilde{x}$ and usually use $C,x$ to refer to the layer-respecting circuits and layered inputs. We inherit the interpretation of $[wnt]$ as consisting of $t$ layers, and continue to use the notation $[wnt] = \bigsqcup_{t'=1}^t T^{(t')} = \bigsqcup_{t'\in [t], j\in [n]} T^{(t')}_j$. Finally, we will typically use $b\in \{\pm 1\}^m$ to denote a string on the ``right hand side''.

\paragraph*{Satisfying constraints as a correlation bound.} With this in mind, let $b\in \{\pm 1\}^{m}$ be a Boolean string claimed to $(\frac{1}{2}-\eps)$-close of $C$. This means there is an input $x\in \{\pm 1\}^{tnw}$ such that
\begin{align}
    \left\langle C(x), b \right\rangle := \mathbb{E}_{\bm{i}\sim [m]}\left[ C(x)_{\bm{i}}\cdot b_{\bm{i}} \right] \ge 2\eps. \label{equ: perfect correlation with b}
\end{align}
Use $C^i(x)$ to denote $C(x)_i$ for every $i \in [m]$. Given $C^i$, we may look at its Fourier expansion:
\begin{align*}
    C^i(x) = \sum_{\alpha\subseteq [ntw]} \widehat{C^i_\alpha} \cdot x_\alpha
\end{align*}
where we use the convention that $x_{\alpha} = \prod_{i:\alpha_i = 1}x_i$.

To proceed, we need to make a few observations about the Fourier expansion.

\begin{itemize}
    \item The total $\ell_1$-mass of $C^i$ on the Fourier basis is bounded by $2^{tw}$ and there can be at most $4^{tw}$ non-zero characters. This follows because $C^i$ can be written as the sum of $2^{tw}$ conjunctions. Note that each conjunction has an $\ell_1$-mass of at most $1$ and Fourier sparsity at most $2^{tw}$.
    \item Every $\alpha$ associated with a non-zero $\widehat{C^i_\alpha}$ has a nice structure: it touches at most one group from each layer. Consequently, the Fourier degree of $C^i$ is bounded by $tw$.
\end{itemize}

Now, we use the linearity of expectation to expand \Cref{equ: perfect correlation with b}.
\begin{align}
2\eps \le \mathbb{E}_{\bm{i}\sim [m]}\left[ C^{\bm{i}}(x)\cdot b_{\bm{i}} \right] = \mathbb{E}_{\bm{i}\sim [m]} \left[ \sum_{\alpha} \widehat{C^{\bm{i}}_\alpha}\cdot x_\alpha \cdot b_{\bm{i}}\right]. \label{equ: correlation with b - expanded}
\end{align}

Given that $|\alpha|\le tw$ is always true, Equation~\ref{equ: correlation with b - expanded} starts to look like a semi-random $(tw)$-XOR instance, with the caveat that every $b_i$ corresponds to many terms, and the arity bound of $tw$ is a bit off from what we desire (i.e.,~$t$).

\paragraph*{Grouping of Fourier characters.} The final step is therefore a grouping manipulation, which splits the Fourier characters into a small number of collections and makes each collection look more like a standard $t$-XOR instance.

In particular, for every $t$-tuple $\beta_1,\beta_2,\dots,\beta_t$ where each $\beta_{t'}$ is a subset of $[w]$, we consider the collection of characters
\begin{align*}
    \mathcal{T}_{\beta_1,\dots, \beta_t} = \{ \alpha\subseteq [ntw]: \forall t'\in [t], \alpha \cap T^{(t')} = \alpha\cap T^{(t')}_j = \beta_{t'} \text{ for some $j\in [n]$} \}.
\end{align*}
Note that we slightly abuse notation by identifying $T^{(t')}_j$ (a size-$w$ subset of $[wnt]$) with the set $[w]$, and compare $\alpha \cap T^{(t')}_j$ with $\beta_{t'}$. In words, recall that for all $\alpha$ of interest to us, it intersects with each layer on at most one group. $\mathcal{T}_{\beta_1,\dots, \beta_t}$ just collects all those $\alpha$ such that the intersection of $\alpha$ to the $t'$-th layer ``looks like'' $\beta_{t'}$ (for every $t'\in [t]$).

Although $\mathcal{T}_{\beta_1,\dots,\beta_t}$'s (when $\beta_1,\dots, \beta_t$ run over all possible tuples) do not cover all Fourier characters of $[wnt]$, they do cover all the non-zero Fourier characters that may appear in a layer-respecting circuit. Hence, we can expand \Cref{equ: correlation with b - expanded} further by writing
\begin{align}
    2\eps \le \mathbb{E}_{\bm{i}\sim [m]} \left[ \sum_{\alpha} \widehat{C^{\bm{i}}_\alpha}\cdot x_\alpha \cdot b_{\bm{i}}\right] = \sum_{\beta_1,\dots, \beta_t} \mathbb{E}_{{\bm{i}}\in [m]}\left[ \sum_{\alpha\in \mathcal{T}_{\beta_1,\dots, \beta_t}} \widehat{C^{\bm{i}}_\alpha} \cdot x_\alpha \cdot b_{\bm{i}} \right]. \label{equ: correlation with b - grouped}
\end{align}

We are still a tiny step away from $t$-XOR instances: the issue is that in \Cref{equ: correlation with b - grouped}, there may be multiple non-zero $\widehat{C}^i_\alpha$ with $\alpha\in \mathcal{T}_{\beta_1,\dots, \beta_t}$. Nevertheless, observe that the number of non-zero $\widehat{C}^i_\alpha$ with $\alpha \in \mathcal{T}_{\beta_1,\dots, \beta_t}$ do not exceed $2^{tw}$. For each $i$ and $\beta = (\beta_1,\dots, \beta_t)$, we can arbitrarily order all non-zero characters $\widehat{C}_{\alpha}^i$ with $\alpha\in \mathcal{T}_\beta$, and list them as $\widehat{C}^i_{\alpha(\beta, i, 1)},\dots, \widehat{C}^i_{\alpha(\beta, i, 2^{tw})}$ (if the number of non-zero characters is less than $2^{tw}$, we append arbitrary valid characters with zero coefficients). Overall, we arrive at
\begin{align}
    2\eps \le \eqref{equ: correlation with b - grouped} =  \sum_{\beta} \sum_{\ell=1}^{2^{tw}} \E_{\bm{i}\sim [m]}\left[ \widehat{C}^{\bm{i}}_{\alpha(\beta,\bm{i}, \ell)} \cdot x_{\alpha(\beta,\bm{i},\ell)} \cdot b_{\bm{i}} \right]. \label{equ: correlation with b - final}
\end{align}
Now, it is clear that for every $\beta,\ell$, the inner term of \Cref{equ: correlation with b - final} can be understood as a (weighted) semi-random $t$-XOR instance. In slightly more detail, any $x\in \{\pm 1\}^{wnt}$ induces a valuation string $y\in \{\pm 1\}^{nt}$ by $y_{t',j} = x_{T^{(t')}_j \cap \beta_{t'}}$ for $t'\in [t],j\in [n]$. We then interpret each $\widehat{C}^i_{\alpha(\beta, i,\ell)} \cdot x_{\alpha(\beta,i,\ell)} \cdot b_i$ as a constraint $y_{1,j_1} y_{2,j_2} \cdots  y_{t,j_t} = b_i$ where $(j_1,\dots, j_t)$ are the relevant groups from each layer that $\alpha(\beta,i,\ell)$ touched upon.

Our final remark is that all the decompositions we have done so far do not really depend on the right-hand side of $b$, and they can be done algorithmically.

Since there are $4^{tw}$ many tuples of $(\beta, \ell)$ from the outer summation of \Cref{equ: correlation with b - final}, averaging principle promises that one of the correlations would be at least $4^{-tw}$. Hence, to refute the claim that $b\in \mathrm{range}(C)$, we need to prove that none of the $t$-XOR instances have value $2\eps \cdot 4^{-tw}$.

\section{Pseudorandom Distributions} \label{sec: pseudorandomness}
In this section, we present the formal definitions of the pseudorandom distributions used throughout the paper.

\begin{definition}[$k$-wise independence]
We say distribution $\cD$ over $\Sigma^n$ is $k$-wise independent if for all $I\subseteq [n]$ of size $|I|=k$ and $a_I\in \Sigma^k$, $\Pr_{\bm{x}\sim \cD}[\bm{x}_I=a_I]=\frac{1}{|\Sigma|^{k}}$.
\end{definition}
\begin{definition}[$\eta$-almost $k$-wise independence]
We say distribution $\cD$ over $\Sigma^n$ is $\eta$-almost $k$-wise independent if for all $I\subseteq [n]$ of size $|I|=k$ and $a_I\in \Sigma^k$, $\frac{1-\eta}{|\Sigma|^{k}}\le \Pr_{\bm{x}\sim \cD}[\bm{x}_I=a_I]\le \frac{1+\eta}{|\Sigma|^{k}}$.
\end{definition}

\begin{definition}[$\eta$-bias]
We say a distribution $\cD$ over $\ZO^n$ is $\eta$-biased if for all $I\subseteq [n]$, $\frac{1}{2}(1-\eta)\le \Pr_{\bm{x}\sim \mathcal{D}}[\bigoplus_{i\in I} \bm{x}_i=0]\le \frac{1}{2}(1+\eta)$.
\end{definition}

\begin{definition}[$k$-wise $\eta$-bias]
We say a distribution $\cD$ over $\ZO^n$ is $k$-wise $\eta$-biased if for all $I\subseteq [n]$ of size $|I|\le k$, $\frac{1}{2}(1-\eta)\le \Pr_{\bm{x}\sim \mathcal{D}}[\bigoplus_{i\in I} \bm{x}_i=0]\le \frac{1}{2}(1+\eta)$.
\end{definition}

\Cref{tab:pseudorandomness} summarizes the number of bits needed to \emph{explicitly} construct the above distributions.
We refer the readers to~\cite{TCS-109} for details.

\begin{table}[h!]
\centering
\renewcommand{\arraystretch}{2}
\begin{tabular}{c|c|c|c}
\hline
\bf Pseudorandomness & \bf Alphabet & \bf Seed length & \bf References \\
\hline
$k$-wise independence & $\Sigma$ & $O(k\log |\Sigma|\cdot \log \frac{n}{k})$ & ~\cite{CGHFRS85,ABI86} \\
$\eta$-almost $k$-wise independence & $\Sigma$ & $O(\log\log n+k\log|\Sigma|+
\log\frac{1}{\eta})$ & \cite{NaorN93, AlonGHP92} \\
$\eta$-bias & \ZO & $O(\log \frac{n}{\eta})$ & \cite{NaorN93} \\
$k$-wise $\eta$-bias & \ZO & $O(\log\log n+\log \frac{k}{\eta})$ & \cite{NaorN93} \\
\hline
\end{tabular}
\caption{Pseudorandom distributions and the seed length of their explicit constructions.}
\label{tab:pseudorandomness}
\end{table}

\section{Trace Method for Refuting Semi-Random $k$-XORs} \label{sec: refute}

The main focus of this section is the following theorem.

\begin{theorem}\label{thm: refute semi random xor}
    There is a large but universal constant $c_{\mathrm{refute}}$ for which the following holds. Let $\mathcal{H}$ be a $k$-uniform hypergraph with $n$ vertices and $m$ edges and $b\in \{\pm 1\}^{\mathcal{H}}$ be a vector. Consider the $k$-XOR instance $\psi_{\mathcal{H}, b}(x) = \frac{1}{|\cH|}\sum_{C\in \mathcal{H}} b_C \cdot \prod_{i \in e_C} x_i $.

    Let $r\ge k$ and $\ell = 2\left\lceil r\log n\right\rceil$. Suppose $\mathcal{D}$ is an $\ell$-wise independent distribution over $\{\pm 1\}^m$ and $m\ge c_{\mathrm{refute}}^k n \log n \left( \frac{n}{r} \right)^{k/2-1} \eps^{-2 - 2\cdot \mathbf{1}\{k \text{ is odd}\}}$. Then, with probability $1-\frac{1}{\poly(n)}$ over $\bm{b}\sim \mathcal{D}$, it holds that
    \begin{align*}
       \val(\psi_{\cH,\bm{b}}) \le \eps,
    \end{align*}
    Furthermore, there is a deterministic algorithm running in time $m+n^{O(r)}$ such that, with probability $1-\frac{1}{\mathrm{poly}(n)}$ over $\bm{b}\sim \mathcal{D}$, the algorithm certifies
    the bound given the input $\psi_{\mathcal{H},\bm{b}}$.
\end{theorem}

\paragraph*{Extensions to weaker pseudo-random sources} For future applications, we will need variants of Theorem~\ref{thm: refute semi random xor} where $\bm{b}$ is drawn from a distribution with even weaker pseudorandom properties. In particular, when $\bm{b}$ is drawn from an $\ell$-independent distribution with $\ell \le \log(n)$, we have:
\begin{theorem}\label{thm: refute semi random xor smaller independence}
    Consider the same setup as Theorem~\ref{thm: refute semi random xor} with an \emph{even} $k$. Suppose instead that $\mathcal{D}$ is an $\ell$-wise independent distribution with $\ell \le \log(n)$. Let $\delta \in (0, 1)$ be a parameter. Suppose now that $m\ge c_{\mathrm{refute}}^k n \log n \cdot (\frac{n}{k})^{k/2-1} \cdot \eps^{-2}\cdot \left( \delta^{-2/\ell} n^{2k/\ell} \right)$. Then, with probability $1-\delta$ over $\bm{b}\sim \mathcal{D}$, it holds that $\val(\cH, \bm{b})\le \eps$ and this bound can be certified by a deterministic algorithm running in time $m + n^{O(r)}$.
\end{theorem}

In Theorem~\ref{thm: refute semi random xor smaller independence}, we only stated the result for the even-arity case. A similar theorem for the odd-arity case holds true, albeit with a slightly larger blow-up in terms of $m$ (i.e., $n^{4k/\ell}$ instead of $n^{2k/\ell}$). Proving this, however, really requires us to examine the arguments of \cite{HsiehKM23} more carefully, and tune the parameters therein. Since we will present a self-contained proof for the case of even $k$, we are able to offer a proof of Theorem~\ref{thm: refute semi random xor smaller independence} with minimal additional effort, allowing us to recover the main result of \cite{KPI25} completely. As the extremely limited independence regime is not the primary focus of this work, we will not present the details for odd $k$'s. From the pseudorandomness perspective, if one is interested in optimizing the number of bits used to sample from $\mathcal{D}$, Theorem~\ref{thm: refute semi random xor almost independence} (appearing next) turns out to provide a better solution.

\paragraph*{Almost-independent sources.} Perhaps the more important extension is to replace the (perfect) $\ell$-wise independent distribution with $\eta$-almost $\ell$-wise independent distributions, where $\eta > 0$ is sufficiently small.
Indeed, we will work on $\ell$-wise $\eta$-biased distributions, which generalize $\eta$-almost $\ell$-wise independent distributions and $\eta$-biased distributions at the same time.

Below we only state and prove the extension for Theorem~\ref{thm: refute semi random xor}. A version of the statement holds true for Theorem~\ref{thm: refute semi random xor smaller independence}, but it will not be needed for our purpose.

\begin{theorem}\label{thm: refute semi random xor almost independence}
    There is a universal constant $c_{\mathrm{almost}}$ for which the following holds. Consider the same setup as in Theorem~\ref{thm: refute semi random xor}. Suppose instead that $\mathcal{D}$ is an $\ell$-wise $\eta$-biased independent distribution. Then, provided the additional assumption that $\eta \le n^{-r} \cdot (c_{\mathrm{almost}} \cdot \eps)^{\ell}$, the same conclusion as in Theorem~\ref{thm: refute semi random xor} holds.
\end{theorem}

\paragraph*{Organization of the proofs.} The proof of Theorem~\ref{thm: refute semi random xor} will mostly follow the reasoning of \cite{HsiehKM23}. Our only new ingredient is that, since the proof of \cite{HsiehKM23} uses a trace method, and it utilizes the randomness of $b$ only up to ``locality'' $\ell$, replacing the truly random $b$ with a $\ell$-wise independent $b$ is sufficient to establish the same conclusion.


In the following, we will reproduce \cite{HsiehKM23}'s proof of Theorem~\ref{thm: refute semi random xor} for the case of even $k$ as it is simple and clean enough. Having presented the details for the even-arity case, we will only outline the proof for the odd-arity case and omit the involved combinatorial manipulations and arguments from \cite{HsiehKM23}, since our simple observation and modification do not change the reasoning along those lines. We explain how to obtain the extensions mentioned above in \Cref{sec: refute extension}. Finally, we give a quick proof of \Cref{thm: main odd arity} at \Cref{sec: refute proof of main} by assembling \Cref{lemma: XOR reduction} and \Cref{thm: refute semi random xor almost independence}.

\subsection{Proof for the Even-Arity Case} \label{sec: refute even}

In this section, we will present a largely self-contained exposition of the semi-random $k$-XOR refutation algorithm of \cite{HsiehKM23} for the even-arity case.

\subsubsection{Construction of Kikuchi Matrices} \label{sec: refute even construction}

Let $k$ be even and $\psi_{\mathcal{H},b}$ a $k$-XOR instance with $n$ variables and $m$ constraints. Let $r\ge \frac{k}{2}$ be a parameter. The level-$r$ Kikuchi matrix $A\in \mathbb{R}^{\binom{n}{r}\times \binom{n}{r}}$ for $\psi_{\mathcal{H},b}$ is constructed as follows: we index rows and columns of $A$ by $r$-subsets of $[n]$ naturally. Then, for every $S,T\in \binom{[n]}{r}$, define $A_{S,T} = \sum_{C\in \mathcal{H}: e_C = S \oplus T} b_C$, where $S\oplus T$ is the symmetric difference of $S$ and $T$. 
Note that when $\mathcal{H}$ does not contain multiple edges among any $k$ variables, $A$ can be understood as the signed adjacency matrix of a graph with vertex set $\binom{[n]}{r}$. The edges exist between pairs $(S,T)$ such that $S\oplus T = e_C$ for $C\in \mathcal{H}$ and they are signed as $b_C$. When $\mathcal{H}$ does contain multiple edges, the induced graph over $\binom{[n]}{r}$ might also contain multiple edges. In any case, we will call the induced graph the level-$r$ Kikuchi graph for $\psi$.

The connection between the Kikuchi matrix $A$ and the $k$-XOR instance $\psi$ is explained as follows. For every assignment $x\in \{\pm 1\}^n$ define its $r$-fold tensorization (without replacement) to be a vector $x^{\odot r}\in \{\pm 1\}^{\binom{n}{r}}$. Here the coordinates of $x^{\odot r}$ are indexed by $r$-subset of $[n]$ and we define $x^{\odot r}_{T} = \prod_{i\in T} x_i$ for every $T\in \binom{[n]}{r}$. Then, we see that
\begin{align}
    (x^{\odot r})^\top A (x^{\odot r}) = \binom{k}{k/2}\binom{n-k}{r-k/2}\sum_{C\in\mathcal{H}} b_C x_{e_C}. \label{equ: kikuchi form to xor value even form 1}
\end{align}
To see why the above equation holds, note that each edge $C\in \mathcal{H}$ contributes to $\binom{k}{k/2}\binom{n-k}{r-k/2}$ pairs of $(S,T)$ in the matrix $A$ (to form $S$, once first chooses $k/2$ items from $e_C$, and then $r-k/2$ items from $[n]\setminus e_C$. Once $S$ is fixed $T$ is determined from $S$ and $C$).

Therefore, the question of upper-bounding $\mathrm{val}(\psi)$ can be reduced to upper-bounding 
\begin{align*}
    \|A\|_{\infty \to 1} := \max_{x,y\in \{\pm 1\}^{\binom{n}{r}}} \{ x^\top A y \}.
\end{align*}
We further reduce the question to analyzing the \emph{spectral norm} of a re-weighted version of $A$. Toward that goal, let $d$ be the average degree of the vertices of the Kikuchi graph. We have
\begin{align*}
    d = \frac{m\binom{k}{k/2}\binom{n-k}{r-k/2}}{\binom{n}{r}} \ge \frac{1}{2}\left( \frac{r}{n} \right)^{k/2} m,
\end{align*}
provided that $k\le r\le n/8$. 

Let $D$ now be the diagonal graph recording the degrees of the vertices of the Kikuchi graph. Namely, for every $S$, let $D_{S,S}$ equal to the number of pairs $(C, T)$ such that $S\oplus T = e_C$. We have $\mathrm{trace}(D) = d\cdot \binom{n}{r}$. Let $\Gamma = D + d\cdot I$ (here $I$ is the identity matrix) and define
\begin{align*}
    B = \Gamma^{-1/2} A \Gamma^{-1/2}.
\end{align*}
Let $x,y\in \{\pm 1\}^{\binom{n}{r}}$. One immediately sees that
\begin{align*}
    x^\top A y = (x^\top \Gamma^{1/2}) B (\Gamma^{1/2} y) \le \| B\|_{2\to 2} \cdot \|\Gamma^{1/2} x\|_2 \|\Gamma^{1/2} y\|_2 = \|B\|_{2\to 2} \cdot \mathrm{trace}(\Gamma) = 2 d \binom{n}{r} \|B\|_{2\to 2}.
\end{align*}
We thus have the upper bound 
\begin{align}
    \| A \|_{\infty \to 1} \le 2d\binom{n}{r} \|B\|_{2\to 2} \ . \label{equ: standard vs reweighted even}
\end{align}
This together with \Cref{equ: kikuchi form to xor value even form 1} implies:
\begin{align}
\psi(x) = \frac{1}{\binom{n}{r} d} (x^{\odot r})^{\top} A (x^{\odot r}) \le 2 \|B\|_{2\to 2}. \label{equ: kikuchi form to xor value even form 2}
\end{align}

\subsubsection{Analysis of Kikuchi Matrices} \label{sec: refute even analysis}

We have shown that, in order to bound $\psi(x)$ and $\|A\|_{\infty \to 1}$, it suffices to upper bound the spectral norm of $B$. 
Towards this, we will prove the following claim.

\begin{lemma}\label{lemma: spectral bound for random b even}
    Let $n, m, r, k$ be parameters such that $k\le r$ and $k$ is even. Consider a semi-random $k$-XOR instance $\psi_{\mathcal{H},\bm{b}}$ with $n$ variables and $m = |\mathcal{H}|$ constraints. Let matrices $A,B$ be the level-$r$ (standard and re-weighted) Kikuchi matrices constructed as above from $\psi_{\mathcal{H},\bm{b}}$. Let $\ell = 2\left\lceil r \log n \right\rceil$ and suppose $\bm{b}$ is randomly drawn from an $\ell$-wise independent distribution $\mathcal{D}$.  Then, it holds that
    \begin{align*}
        \Pr_{\bm{b}\sim \mathcal{D}}\left[ \|B\|_{2\to 2}\ge \sqrt{\frac{r\log n}{d}} \right]\le \frac{1}{n^{100}}.
    \end{align*}
\end{lemma}

\begin{proof}
We use a trace method. Since $\ell$ is even, it is well-known that (given $B$ is a symmetric matrix)
\begin{align*}
    \|B\|_{2 \to 2} = \| B^\ell \|_{2 \to 2}^{1/\ell} \le \mathrm{trace}(B^\ell)^{1/\ell} = \mathrm{trace}((\Gamma^{-1} A)^{\ell})^{1/\ell}.
\end{align*}
Raising the equation to the $\ell$-th power and taking expectation over $\bm{b}$, we obtain
\begin{align*}
    \E_{\bm{b}}[ \|B\|_{2\to 2}^{\ell}] \le \E_{\bm{b}}[ \mathrm{trace}((\Gamma^{-1} A)^{\ell})].
\end{align*}
Here $\mathrm{trace}((\Gamma^{-1} A)^{\ell})$ can be understood as counting the (weighted) number of closed walks of length $\ell$ in the graph described by $\Gamma^{-1}A$. Fix $S_1\to^{C_1} S_2 \to^{C_2} \dots \to^{C_{\ell-1}} S_{\ell} \to^{C_\ell} S_1$ to be one such walk. We use notation $S_i \to^{C} S_{i+1}$ to denote that the edge used from $S_i$ to $S_{i+1}$ is $C\in \mathcal{H}$. The crucial observation is that, if there is some $C\in \mathcal{H}$ that is used an \emph{odd} number of times in the walk, then the average contribution of this walk (over the randomness of $\bm{b}$) is zero. On the other hand, if every edge $C\in \mathcal{H}$ is used an \emph{even} number of times, the contribution of this walk is simply $\prod_{i=1}^{\ell} \Gamma_{S_i,S_i}^{-1}$ regardless of the realization of $\bm{b}$.

The above observation dictates that every walk with non-zero contribution involves at most $\ell/2$ hyperedges. Hence, a \emph{necessary} condition for a walk $(S_1,\dots, S_{\ell})$ to have non-zero contribution is that $|\{S_i\}_{i\in [1,\ell]}|\le \frac{\ell}{2}$ and $S_i \in N(S_{i-1})$ for every $i$, where $N(S_{i-1})$ denotes the neighbors of $S_{i-1}$. With this in mind, we can write
\begin{align}
     &~~~~ \E_{\bm{b}}[ \mathrm{trace}((\Gamma^{-1} A)^{\ell})] \notag \\
     &\le \sum_{S_1} \frac{1}{\Gamma_{S_1,S_1}} \sum_{S_2\in N(S_1)} \frac{1}{\Gamma_{S_2,S_2}} \dots \sum_{S_{\ell}\in N(S_{\ell-1})} \frac{\mathbf{1}\{ | \{S_i\}_{i\in [1,\ell]}| \le \ell/2 \land S_{1} \in N(S_\ell)\}}{\Gamma_{S_{\ell-1},S_{\ell-1}} \cdot \Gamma_{S_{\ell},S_{\ell}}}.\label{equ: trace method even step 1}
\end{align}
Let us inspect $|\{S_i\}_{1\le i\le t}|$ as $t$ increases. See that $S_1$ is certainly a new addition to the collection and there are $\binom{n}{r}$ possible choices of $S_1$. Other than $S_1$, we enumerate a size-$(\ell/2-1)$ subset of $[2,\ell]$, denoted by $T\subseteq [2,\ell]$, and we think of $T$ as those steps where there \emph{may}\footnote{It is possible that the number of new additions is less than $\ell/2-1$. However, as it will be clear from the proof, only enumerating $T$ of size $\ell/2-1$ suffices to cover all contributions.} be a new addition to the set $\{S_i\}_{1\le i\le t}$. For every $i\in T$, the choice of $S_{i}$ (having chosen $S_{i-1}$) has $D_{S_i,S_i}$ possibilities and their weighted contribution is $\frac{D_{S_i,S_i}}{\Gamma_{S_i,S_i}}\le 1$. For every $i\notin T$, the choice of $S_i$ has at most $i-1\le \ell$ possibilities and their contribution is at most $\frac{\ell}{\Gamma_{S_i,S_i}}\le \frac{\ell}{d}$. This allows us to upper-bound
\begin{align}
    \eqref{equ: trace method even step 1}\le \binom{n}{r} \cdot \sum_{T} 1^{\ell/2-1} \cdot (\frac{\ell}{d})^{\ell/2} \le \binom{n}{r} \binom{\ell-1}{\ell/2} (\frac{\ell}{d})^{\ell/2} \le n^r 2^\ell \left( \frac{\ell}{d} \right)^{\ell/2} \le O\left( \frac{\ell}{d} \right)^{\ell/2}. \label{equ: trace method even step 2}
\end{align}
Finally, we may use Markov's inequality to conclude that 
\begin{align*}
    \Pr_{\bm{b}}\left[ \|B\|_{2\to 2}\ge \lambda \right] \le \lambda^{-\ell}\cdot \left( n^r 2^\ell \left( \frac{\ell}{d} \right)^{\ell/2} \right) \le O\left( \frac{\ell}{\lambda^2 d} \right)^{\ell/2}.
\end{align*}
Setting $\lambda = O(\sqrt{\ell/d})$ establishes the lemma.
\end{proof}

\subsubsection{Wrap-up for the Even-Arity Case} \label{sec: refute even summary}

We now conclude the proof of Theorem~\ref{thm: refute semi random xor} for the case of even $k$. Namely, whenever $\bm{b}$ is such that the conclusion of Lemma~\ref{lemma: spectral bound for random b even} holds (which happens with probability $1-\frac{1}{\mathrm{poly}(n)}$ if drawn from an $\ell$-wise independent distribution), we have (by \Cref{equ: standard vs reweighted even}):
\begin{align}
    \max_{x\in \{\pm 1\}^n} \left\{ \psi(x) \right\} \le \frac{1}{\binom{n}{r} d} \| A \|_{\infty \to 1} \le 2 \|B\|_{2\to 2} \le 2 \sqrt{\frac{r\log n}{d}} \le \eps. \label{equ: val bounded by kikuchi even}
\end{align}
To verify the last inequality, recall that $d\ge \frac{1}{2}\left(\frac{r}{n}\right)^{k/2}m \ge c'\cdot r\log n\cdot \eps^{-2}$ where $c'$ is a large constant depending on $c_{\mathrm{refute}}$.


\subsection{Proofs of the Odd-Arity Case}

The proof for the odd-arity case involves more manipulation. Here, we outline only the key steps in the proof and refer the reader to \cite{HsiehKM23} for more details.

\subsubsection{Construction of the Kikuchi Matrices}\label{sec: construct kikuchi odd k}

Let $k\ge 3$ be odd and $\psi_{\mathcal{H},b}$ be a $k$-XOR instance. For any level $r\ge \frac{k+1}{2}$, it is shown in \cite{HsiehKM23} (following \cite{GKM}) how to split $\mathcal{H} = \mathcal{H}^1\sqcup \dots \sqcup \mathcal{H}^{k-1}$ and construct a sequence of $k-1$ matrices $A^{(1)}_b,\dots, A^{(k-1)}_b$, with the following properties:
\begin{itemize}
    \item The instance $\psi_{\mathcal{H},b}$ is decomposed into $(k-1)$ sub-instances $\psi^{1}_{\mathcal{H},b},\dots, \psi^{k-1}_{\mathcal{H},b}$, where $\psi^t_{\mathcal{H},b}(x) =\E_{\bm{C}\sim \mathcal{H}^t}[b_{\bm{C}}\cdot x_{e_{\bm{C}}}]$. Each $A_b^{(t)}$ is ``responsible'' for one instance and is a matrix in $\mathbb{R}^{\binom{2n}{r}\times \binom{2n}{r}}$. The rows and columns of $A_b^{(t)}$ are indexed by size-$r$ subsets of $2n$ naturally.
    \item The entries of all $A^{(t)}_b$ are determined as follows. For every pair of distinct edges $C,C'\in \mathcal{H}^{t}$, either they do not contribute to any matrix, or there is some $t\in [k]$ and a subset $E_{C,C'} \subseteq \{ (S,T) \in \binom{2n}{r}\times \binom{2n}{r} \}$ of entries of $A^{(t)}_b$ which receive a contribution of $b_{C}\cdot b_{C'}$ from the pair $(C,C')$. 
\end{itemize}

\paragraph*{Poniters to HKM.} In the above, we summarized and distilled the key constructions from \cite{HsiehKM23} that are most essential to us. If readers are also interested in the details of the construction, we provide a brief road map to relevant sections of \cite{HsiehKM23} for reference.

The first step is to distribute the hyperedges of $\mathcal{H}$ into $k-1$ hypergraphs $\mathcal{H}^{(1)},\dots, \mathcal{H}^{(t-1)}$ according to certain rules. The readers are referred to Algorithm 2 in Section 4.2 of \cite{HsiehKM23} for the details. After the decomposition, the $t$-th matrix $A_b^{(t)}$ will be constructed from $\mathcal{H}^{(t)}$. 

Next, consider the most natural way to construct a Kikuchi matrix (if we were to follow the approach of the even-arity case): one would collect subsets of $[n]$ of a chosen size, and add an edge between $(S,T)$ if $S\oplus T = e_C$ for some $C\in \mathcal{H}$. Unfortunately, since $e_C$ is of odd size, this naive construction necessitates an imbalanced Kikuchi graph (with subsets of different sizes) and results in suboptimal bounds. It turns out a better idea is to add an edge between $(S,T)$ if (roughly speaking) $S\oplus T = e_C\oplus e_{C'}$ for some $C,C'\in \mathcal{H}$. However, implementing this plan requires care and creativity. Refer to Section 3 and Section 4.2 of \cite{HsiehKM23} for details.


\subsubsection{Analysis of the Kikuchi Matrices}

Next, we state the following advanced properties about these matrices. The first claim connects the value of the XOR-instance with quadratic forms of the Kikuchi matrices.
\begin{claim}\label{claim: value as quadratic form odd k}
There is an absolute constant $c_{\text{odd}}$ for which the following is true. Let $k\ge 3$ be odd and $r\ge 1$ be a parameter (the ``Kikuchi level''). Let $\psi_{\mathcal{H},b}$ be a $k$-XOR instance, with $n$ variables and $m = |\mathcal{H}| \ge c_{\text{odd}}^k\cdot  n \log n (\frac{n}{r})^{\frac{k}{2} - 1} \eps^{-4}$ clauses. Suppose further that $2k\le r\le n$.

Let $\psi_{\mathcal{H},b}^1,\dots, \psi_{\mathcal{H},b}^{k-1}$ 
be sub-instances of $\psi_{\mathcal{H},b}$ and $A_{b}^{(1)},\dots, A_{b}^{(k-1)}$ be the Kikuchi matrices from Section~\ref{sec: construct kikuchi odd k}. For any $x\in \{\pm 1\}^{n}$, define $x^{\odot r} \in \{\pm 1\}^{\binom{2n}{r}}$ such that for $S\subseteq [n]\times [2]$ with $S = (S^{(1)},S^{(2)})$, the $S$-th entry of $x^{\odot r}$ is $x_{S^{(1)}\oplus S^{(2)}}$. Then, it holds that
$$
\forall t\in [1,k-1], ~~ \psi^t_{\mathcal{H},b}(x)^2 \le \frac{\eps^2}{2} + \frac{k p_t}{2 \alpha_t m^2} (x^{\odot r})^\top A_{b}^{(t)} x^{\odot r},
$$
where $\alpha_t, p_t$ are parameters chosen the same as in \cite{HsiehKM23}. 
\end{claim}

The parameters are taken from \cite[Lemma 4.5]{HsiehKM23} verbatim (c.f.~\cite[Equation (5), (6) and (8)]{HsiehKM23}. For the proof of this claim, the readers are referred to \cite[Proof of Lemma 4.5]{HsiehKM23}.

The next claim shows that when $\bm{b}$ is chosen randomly from a limited-independence source, we have the desired upper bound on $\psi^t_{\mathcal{H},\bm{b}}(x)^2$.

\begin{claim}\label{claim: value upper bound for odd k}
    Consider the same setup as in Claim~\ref{claim: value as quadratic form odd k}. Let $\ell = 2\lceil r \log n\rceil$. Suppose $\bm{b}$ is chosen from an $2\ell$-independent distribution. Then, for every $t\in [k-1]$ and $\psi_{\mathcal{H},\bm{b}}^t$, there is an efficiently computable (in time $n^{O(r)}$) quantity $\eta(\mathcal{H}^t,\bm{b})$ such that, with probability $1-\frac{1}{\mathrm{poly(n)}}$, we have
    \begin{align*}
        \max_{x\in \{\pm 1\}^{n}} {(\psi^t_{\mathcal{H},\bm{b}}(x))}^2 \le \eta(\mathcal{H}^t,\bm{b}) \le  \eps^2.
    \end{align*}
\end{claim}

\begin{proof}[Proof sketch]
    When $\bm{b}$ is chosen from the uniform distribution, it was shown in \cite[Proof of Lemma 4.5]{HsiehKM23} that the claimed bound is true. Inspecting their proof, one can find that the randomness of $\bm{b}$ is used only to give an upper bound for $\E_{\bm{b}}[ \mathrm{tr}( (\Gamma^{-1} \widetilde{A}_{\bm{b}})^{\ell} )]$ where $\Gamma, \widetilde{A}_{\bm{b}}$ are induced from $A^{(t)}_{\bm{b}}$'s. Specifically, $\Gamma$ is a degree matrix independent of $\bm{b}$, and $\widetilde{A}_{\bm{b}}$ is a slight modification of $A^{(t)}_{\bm{b}}$. Furthermore, from the spectral norm bound of $\| \Gamma^{-1} \tilde{A}_{\bm{b}} \|_2$ (which can be computed in $n^{O(r)}$ time), one can compute a quantity $\eta(\mathcal{H}^t, \bm{b})$ that gives an upper bound on $\psi^t(x)$ (this is through the connection from Claim~\ref{claim: value as quadratic form odd k}). As it turns out, we have that $\eta(\mathcal{H}^t, \bm{b})$ is upper-bounded by $\eps^2$ with high probability over a random $\bm{b}$.
    
    We now explain why a limited-independent source of $\bm{b}$ suffices. Let us expand the terms in $\mathrm{tr}[(\Gamma^{-1} \widetilde{A}_{\bm{b}})^\ell )]$. We see that each term naturally corresponds to a length-$\ell'$ closed walk in the (weighted) graph $\Gamma^{-1}\widetilde{A}_{\bm{b}}$. Focusing on one such walk, we observe that if some term $b_C$ appears an \emph{odd} number of times in the walk, then the average contribution of this walk is zero when one takes expectation over $\bm{b}$. Thus, the quantity $\mathbb{E}_{\bm{b}}[\mathrm{tr} ((\Gamma^{-1} \widetilde{A}_{\bm{b}})^\ell ))]$ can be understood as counting the (weighted) number of length-$\ell'$ closed walks where each $b_C$ is used an \emph{even} number of times. This key observation (combined with additional combinatorial properties obtained from the manipulation of \cite{HsiehKM23} we have briefly mentioned in Section~\ref{sec: construct kikuchi odd k}) allows \cite{HsiehKM23} to obtain the claimed bound. 
    
    Finally, our simple observation is that, since each walk of $(\Gamma^{-1} \widetilde{A}_{\bm{b}})^{\ell}$ 
    involves at most $\ell$ hyper-edges, drawing $\bm{b}$ from a $\ell$-independent source suffices to ``fool'' the proof and establish the same conclusion.
\end{proof}

Combining Claims~\ref{claim: value as quadratic form odd k} and \ref{claim: value upper bound for odd k}, the proof of Theorem~\ref{thm: refute semi random xor} is complete for the case of odd $k$.

\subsection{Extensions} \label{sec: refute extension}

We now establish the extensions of Theorem~\ref{thm: refute semi random xor} as promised.

\begin{proof}[Proof of Theorem~\ref{thm: refute semi random xor smaller independence}]
We construct the Kikuchi graphs as in Section~\ref{sec: refute even construction}, with Kikuchi level $r = k$. We obtain that the average degree of the Kikuchi graph will be $d \ge \frac{1}{2}\left( \frac{r}{n} \right)^{k/2} \cdot m$.

We then run the analysis of Section~\ref{sec: refute even analysis}. This time though, we will only be able to raise the matrix to $\ell$-th power. Using similar counting argument, we still obtain that
\begin{align*}
    \E_{\bm{b}}[ \|B\|_{2\to 2}^\ell ] \le n^r 2^\ell \left( \frac{\ell}{d} \right)^{\ell/2} \le O\left( \frac{\ell\cdot n^{2k/\ell}}{d} \right)^{\ell/2}.
\end{align*}
We will apply Markov's inequality on $\|B\|_{2\to 2}$ with threshold $\lambda = O(\delta^{-1/\ell} \cdot \sqrt{\frac{\ell\cdot n^{2k/\ell}}{d}})$. This gives us that
\begin{align*}
    \Pr[\|B\|_{2\to 2} \ge \lambda] \le \left( \frac{\ell n^{2k/\ell}}{d \lambda^2 } \right)^{\ell / 2} \le \delta.
\end{align*}
Whenever the spectral upper bound for $\|B\|_{2\to 2}$ holds true, we obtain that $\mathrm{val}(\psi_{\cH,b})\le 2\lambda$ (c.f.~\Cref{equ: val bounded by kikuchi even}). Finally, our assumption on $m$ implies that $d\ge c'\cdot \delta^{-2/\ell}\cdot \eps^{-2} \cdot \ell \cdot n^{2k/\ell}$ where $c'$ is a large constant. Consequently we have that $\lambda \le \frac{\eps}{2}$.
\end{proof}

\begin{proof}[Proof of Theorem~\ref{thm: refute semi random xor almost independence}]
    We just need to ensure that the trace method is fooled by an $\ell$-wise $\eta$-biased distribution. Here we only argue it for the even-arity case, and the proof for the odd-arity case is similar. 

    Consider \Cref{equ: trace method even step 1} but now with a $\bm{b}$ drawn from an $\ell$-wised $\eta$-biased independent source. Comparing this with the case of a perfect independent source, we find that the contribution of each closed walk changes by at most $\eta$. Although there are much more than $\binom{n}{r}$ walks, we observe that each walk $S_1\to \dots \to S_{\ell}$ comes with a weight of $\prod_{i=1}^{\ell} \frac{1}{\Gamma_{S_i,S_i}}$. It is therefore easy to see that the total weight of all walks is bounded by $\binom{n}{r}$, i.e., the number of ways of choosing the starting point. As such, working with $\ell$-wise $\eta$-biased source only affects the final conclusion of \Cref{equ: trace method even step 2} by an additive term of $\eta\cdot \binom{n}{r} \le (c_{\mathrm{almost}}\cdot \eps)^{\ell} \le \left( \frac{\ell}{d} \right)^{\ell/2}$ (recall that our parameter setting ensures $\ell \approx r\log n$ and $d\approx \frac{r\log n}{\eps^2}$), which can be ``absorbed'' into the big-Oh notation of $O(\ell/d)^{\ell/2}$ in the final bound of \Cref{equ: trace method even step 2}. 

    The case of the odd-arity case is similar. The main property one relies on is that the total weight of all closed walks in the re-weighted Kikuchi graph is bounded by the number of vertices, which in the case of odd-arity happens to be $\binom{2n}{r} \le (2n)^r$. Given this, one wants to choose $\eta$ so that the ``perturbation'' caused by the imperfect independence can be absorbed into the original trace-moment bound with perfect independence. The readers are referred to \cite{HsiehKM23} for the exact calculations.
\end{proof}

\subsection{Proof of \Cref{thm: main odd arity}} \label{sec: refute proof of main}

Equipped with \Cref{lemma: XOR reduction} and \Cref{thm: refute semi random xor almost independence}, we are ready to conclude the proof of \Cref{thm: main odd arity}.

Let $C: \Sigma^n\to \{\pm 1\}^m$ be the input circuit of \Cref{thm: main odd arity}. Apply \Cref{lemma: XOR reduction} on $C$ to produce the layer-respecting circuit $\widetilde{C}$ and the associated ensemble of $t$-XOR schemes $\{ \psi_{\mathcal{H}^\ell, *}\}_{\ell}$. Through appropriate re-parameterization of \Cref{thm: refute semi random xor almost independence}, there is a refutation algorithm that takes input $\mathcal{H}^{\ell}$ and $\bm{b}\sim \mathcal{D}$ ($\mathcal{D}$ is the distribution from the statement of Theorem~\ref{thm: main odd arity}) and produces a certificate of $\val(\psi_{\mathcal{H}^{\ell},\bm{b}}) \le 4^{-tw} \eps$. By union bound, the algorithm succeeds in producing a certificate for every $\psi_{\mathcal{H}^\ell,\bm{b}}$ with probability $1-\frac{1}{\poly(n)}$ over $\bm{b}\sim \mathcal{D}$. Whenever this happens, we have by Item 2 of \Cref{lemma: XOR reduction} that $\Delta(C(x,*),\bm{b})\ge \frac{1}{2} - \eps$ for every $x\in \Sigma^n$, as was to be shown.

\paragraph*{Data structure lower bound.} We now prove \Cref{thm: data structure lb} from \Cref{thm: main odd arity}. Suppose for contradiction that there is a space-$S$ and time-$t$ data structure for $f$. We interpret the query strategy of the data structure as a circuit $C:\Sigma^{S}\to \{0,1\}^m$. It follows that the range of $C$ covers the rows of $f$. Now, take $\mathcal{D}$ to be an $\eta$-almost $k$-wise independent distribution supported on rows of $f$, whose existence is assumed in the theorem statement. We then have that $\mathrm{range}(C) \supseteq \mathrm{supp}(\mathcal{D})$. However this is impossible as we have shown in \Cref{thm: main odd arity} that with high probability over $\bm{b}\sim \mathcal{D}$, the fractional Hamming distance from $\bm{b}$ to any string in the range of $C$ is lower bounded by $\frac{1}{2} - \eps$.

\paragraph*{Lower bound with even smaller independence.} A curious reader may find that, Theorem~\ref{thm: data structure lb} only covers \cite[Item 1 of Theorem 1]{KPI25}. The second item of \cite[Theorem 1]{KPI25} considers distributions with even more limited independence property. We can prove the same item by using \Cref{thm: refute semi random xor smaller independence} in place of \Cref{thm: refute semi random xor almost independence}, to prove a version of \Cref{thm: main odd arity} with limited independence. It is not hard to verify that the condition of $m\approx n^{\frac{t}{2}}\cdot n^{\frac{2t}{\ell}}$ from \Cref{thm: refute semi random xor smaller independence} will translate to a lower bound of $S\ge m^{\frac{1}{t(\frac{1}{2}+\frac{2}{\ell})}}$ (ignoring dependence on minor parameters and log factors), whenever the data structure problem $f\in \{\pm 1\}^{N\times m}$ supports an $\ell$-wise independent distribution, matching the parameter of \cite{KPI25}.


\section{Improved Bounds for Bit-Probe Model}

In this section, we prove~\Cref{thm: adaptive-ds-lb} and~\Cref{thm: nonadaptive-ds-lb}.

\subsection{Nonadaptive Bit-Probe}
We will first prove the simpler \Cref{thm: nonadaptive-ds-lb}. According to~\Cref{sec: refute proof of main}, it suffices to prove the following theorem in analogy to~\Cref{thm: main odd arity} which establishes~\Cref{thm: data structure lb}.

\begin{theorem}\label{thm:nonadaptive-bit-probe}
There are universal constants $c_{\text{bias}},c_{\text{nonadaptive}}$ such that the following holds. Let $t\ge 3,n,m\ge 1$ and $\mathcal{D}$ be a distribution over $\{0,1\}^m$ that is $(c_{\text{bias}}\cdot n)^{-t}$-biased.

Let $C:\ZO^n\to \ZO^m$ be a multi-output circut.
Suppose that each output bit $C(x)_i$ depends on at most $t$ bits of the input $x$. Then provided that $m\ge c^t_{\text{nonadaptive}}\cdot n^{(t-1)/2}\cdot \log n$,
\[
    \Pr_{\bm{b}\sim \mathcal{D}}[\bm{b}\in \range(C)]\le 2/3.
\]
\end{theorem}
\begin{proof}
    Throughout the proof, our analysis will be carried out in the $\{\pm 1\}$-basis.
    For each $i\in [n]$, we can write
    \[
    C(x)_i=g^i(x_{p^i_1},\ldots,x_{p^i_t})=\sum_{\alpha\subseteq [t]} \widehat{g^i_\alpha} \prod_{j\in \alpha} x_{p^i_j}
    \]
    for some $t$-ary predicate $g^i:\{\pm 1\}^t\to \{\pm 1\}$ and indices $p^i_1,\ldots,p^i_t\in [n]$.

    If at least $n+1$ of the predicates are XOR or NXOR, then we can find a collection of indices $i_1,\ldots,i_{r}\in [m]$ and a constant $b^*\in \{\pm 1\}$ such that $\prod_{j=1}^r C(x)_{i_j}=b^*$ for all $x\in \{\pm 1\}^n$.
    Using the assumption that $\mathcal{D}$ is $(c_{\text{bias}}\cdot n)^{-t}$-biased,  we deduce that
    \[
    \Pr_{\bm{b}\sim \mathcal{D}}[\bm{b}\in \range(C)]\le \Pr_{\bm{b}\sim \mathcal{D}}[\prod_{j=1}^r \bm{b}_{i_j}=b^*]\le 1/2+(c_{\text{bias}}\cdot n)^{-t}\le 2/3.
    \]
    
    Let us from now on consider the case that at most $n$ output bits are XOR or NXOR.
    We may further assume that all the output bits are neither XOR nor NXOR by removing such bits and working with the resulting pruned circuit $C':\{\pm 1\}^n\to \{\pm 1\}^{m'}$ with stretch $m'=(1-o(1))m$.
    For each $i\in [m]$, we have
\begin{equation}\label{eqn:bound_leading_term_coefficient}
    \lvert \widehat{g}^i_{[t]}\rvert =2^{-t}\cdot\lvert\sum_{x\in \{\pm 1\}^t} g^i(x)x_{[t]}\rvert \le 1-2^{1-t}
    \end{equation}
    since $g^i$ is neither XOR nor NXOR.
    

    For each $\alpha\subsetneq [t]$, we construct an $|\alpha|$-uniform hypergraph $\cH_{\alpha}\coloneqq \{(p^i_j)_{j\in \alpha}\mid i\in [m]\}$ with $n$ vertices and $m$ edges.
    Let $\widehat{g_\alpha}\coloneqq (\widehat{g^i_\alpha})_{i\in [m]}\in [-1,1]^m$ be a real-valued vector, and
    \[
    \widetilde{B} \coloneqq \{b\in \{\pm 1\}^m\mid \val(\psi_{\cH_\alpha,b\odot\widehat{g_\alpha}})\}\le 2^{-2t}, \forall\alpha\subsetneq [t]\},
    \]
    where $\odot$ denotes the entrywise product.
    We prove that $\widetilde{B}\cap \range(C)=\emptyset$.
    Indeed, fix any $b\in \widetilde{B}$ and observe that for each $x\in \{\pm 1\}^n$,
    \begin{align*}
    \left\langle C(x), b \right\rangle&=\frac{1}{m}\sum_{i=1}^m b_ig^i(x_{p^i_1},\ldots,x_{p^i_t})\\
    &=\frac{1}{m}\sum_{i=1}^m b_i\sum_{\alpha\subseteq [t]} \widehat{g^i_\alpha}\prod_{j\in \alpha} x_{p^i_j}\\
    &=\frac{1}{m}\sum_{\alpha\subseteq [t]}  \sum_{i=1}^m \widehat{g^i_\alpha}\cdot b_i\prod_{j\in \alpha} x_{p^i_j}\\
    &\le \sum_{\alpha\subsetneq [t]} \psi_{\cH_{\alpha},b\odot \widehat{g_\alpha}}(x)+\frac{1}{m}\sum_{i=1}^m|\widehat{g^i_{[t]}}| \\
    &\le (2^t-1)\cdot 2^{-2t}+(1-2^{1-t}) & (b\in \widetilde{B}\text{ and \eqref{eqn:bound_leading_term_coefficient}})\\
    &<1.
    \end{align*}

    For each $\alpha\subsetneq [t]$, by~\Cref{thm: refute semi random xor almost independence},
    we have $\Pr_{\bm{b}\sim \mathcal{D}}[\val(\psi_{H_\alpha,\bm{b}\odot \widehat{g_\alpha}})\ge 2^{-2t}]\le \frac{1}{\poly(n)}$ provided that $m\ge c^t_{\text{nonadaptive}}\cdot n^{(t-1)/2}\cdot \log n\ge c^t_{\text{almost}}\cdot n\log n\cdot (2n/t)^{\frac{t-3}{2}}\cdot (2^{2t})^4$.
    Finally, by applying union bound, we conclude that when $n$ is large enough,
    \[
    \Pr_{\bm{b}\sim \cD}[\bm{b}\in \range(C)]\le 1-\Pr_{\bm{b}\sim \cD}[\bm{b}\in \widetilde{B}]\le (2^t-1)\cdot \frac{1}{\poly(n)}\le 2/3. \qedhere
    \]
\end{proof}

\subsection{Adaptive Bit-Probe}
Now we prove~\Cref{thm: adaptive-ds-lb}, which is implied by the following theorem.

\begin{theorem}\label{thm:adaptive-bit-probe}
Let $c_{\text{bias}}$ be as in~\Cref{thm:nonadaptive-bit-probe}.
For every $t\ge 3$, there is a universal constant \footnote{For the ease of the presentation, we prove the statement with $c_{\text{adaptive}}=2^{2^{O(t)}}$, which suffices for our purpose.
    We remark that $c_{\text{adaptive}}$ can be further reduced to $2^{O(t)}$ through a more refined analysis.} $c_{\text{adaptive}}=c_{\text{adaptive}}(t)$ such that the following holds. Let $n,m\ge 1$ and $\mathcal{D}$ be a distribution over $\{0,1\}^m$ that is $(c_{\text{bias}}\cdot n)^{-t}$-biased.

Let $C:\ZO^n\to \ZO^m$ be a multi-output circut.
Suppose that each output bit $C(x)_i$ can be computed by a $t$-query adaptive decision tree over the input $x$. Then provided that
\[
m\ge \begin{cases}
c_{\text{adaptive}}\cdot n^{t/2-\frac{t-2}{2(t+2)}}\cdot \log^5 n & \text{if $t$ is odd} \\
c_{\text{adaptive}}\cdot n^{t/2-\frac{t-2}{2(t+6)}}\cdot \log^5 n & \text{if $t$ is even}
\end{cases},
\]
we have
\[
    \Pr_{\bm{b}\sim \mathcal{D}}[\bm{b}\in \range(C)]\le  \frac{1}{\mathrm{poly}(n)}.
\]
\end{theorem}

Our proof of~\Cref{thm:adaptive-bit-probe} is a simple adaptation of the proof of Theorem 2 in~\cite[Theorem 2]{KPI25}.
The key observation is that their analysis can be extended to all degree-$t$ predicates with level-$t$ $\ell_1$-Fourier weight bounded by $1$.
We prove that $t$-query adapative decision trees possess this property.
\begin{lemma}\label{lemma: DT_Fourier_weight}
    Let $g:\{\pm 1\}^n\to \{\pm 1\}$ be a Boolean function with Fourier expansion $g(x)=\sum_{\alpha\subseteq [n]} \widehat{g_\alpha}x_\alpha$.
    Suppose $g$ can be computed by a $t$-query adaptive decision tree, then
    \[
    L_{1,t}(g)\coloneqq \sum_{\alpha\subseteq[n]:|\alpha|=t} |\widehat{g_\alpha}|\le 1.
    \]
\end{lemma}
\begin{proof}
We prove the lemma by induction: The statement clearly holds for $t=1$. Now assume the statement is true for $t\le d$.
For any $g:\{\pm 1\}^n\to \{\pm 1\}$ that can be computed by a $(d+1)$-query decision tree $T$, let $x_j$ be the variable queried at the root of $T$.
We can then write
\[
g(x)=\frac{1+x_j}{2} g|_{x_j=1}+\frac{1-x_j}{2}g|_{x_j=-1}.
\]
Observe that both $g|_{x_j=1}$ and $g|_{x_j=-1}$ can be computed by $d$-query decision trees and have degree at most $d$.
By induction hypothesis, we have
\[
L_{1,d+1}(g)\le \frac{1}{2}(L_{1,d}(g|_{x_j=1})+L_{1,d}(g|_{x_j=-1}))\le 1.\qedhere
\]
\end{proof}
We also need the following result from~\cite{KPI25}.
\begin{theorem}[\cite{KPI25}]\label{thm: weak_refutation}
 For every $k\ge 2$, there are universal constants $c_{\text{weak}}=c_{\text{weak}}(k),c_{\text{unsat}}=c_{\text{unsat}}(k)$ such that the following holds.
 Let $\mathcal{H}$ be a $k$-uniform hypergraph with $n$ vertices and $m$ edges, and $r$ be any integer satisfying $k\le r\le n/8$.

    Let $\mathcal{D}$ be a $\gamma$-biased distribution over $\{\pm 1\}^{\cH}$. Then provided that $m\ge c_{\text{weak}}\cdot n(\frac{n}{r})^{k/2-1}\log n$,
    \[
        \Pr_{\bm{b}\sim \cD}[\val(\psi_{\cH,\bm{b}})\le 1-\frac{1}{c_{\text{unsat}}\cdot r\log n}]\ge 1-2\gamma.
    \]
\end{theorem}
    We are now ready to present the proof of~\Cref{thm:adaptive-bit-probe}.
    \begin{proof}[Proof of~\Cref{thm:adaptive-bit-probe}]
    Throughout the proof, we will work under the $\{\pm 1\}$-basis.
    Since a $t$-query adaptive decision tree depends on at most $2^t-1$ variables, 
    we can assume that there is a universal predicate $g:\{\pm 1\}^{2^t-1}\to \{\pm 1\}$ such that for every $i\in [m]$, $C(x)_i=g(x_{p^i_1},\ldots,x_{p^i_{2^t-1}})$ for some $p^i_1,\ldots,p^i_{2^t-1}\in [n]$.
    This assumption can be made without loss of generality since it will only incur a constant multiplicative overhead to the stretch in light of the fact there are at most $(2^t-1)^{2^t-1}\cdot 2^{2^t}=2^{2^{O(t)}}$ different depth-t decision trees.
    
    Let us now write down the predicate in its Fourier expansion form: $g(x)=\sum_{\alpha \subseteq [2^t-1]} \widehat{g_\alpha}x_\alpha$.
    Observe that $g$ has degree at most $t$, namely, $\widehat{g_\alpha}=0$ for all $|\alpha|>t$.
    Moreover, the number of non-zero coefficients of $g$ is at most $\lVert \widehat{g} \rVert_0\le 4^t$.

    Similar to what we have done in the proof of~\Cref{thm:nonadaptive-bit-probe}, for each $\alpha\subseteq [2^t-1]$ of size $|\alpha|\le t$, we construct an $|\alpha|$-uniform hypergraph $\cH_{\alpha}\coloneqq \{(p^i_j)_{j\in \alpha}\mid i\in [m]\}$ with $n$ vertices and $m$ edges.

    Set $r=n^{1/(t+2+4\cdot \mathbf{1}\{t \text{ is even}\})}$.
    Let $\widetilde{B}$ denote the set of all $b\in \{\pm 1\}^m$ such that:
    \begin{enumerate}
        \item 
    $\val(\psi_{\cH_{\alpha},b})\le \frac{1}{2^{2t+1}\cdot c_{\text{unsat}}\cdot r\log n}$ for all $|\alpha|<t$; 
 \item $\val(\psi_{\cH_\alpha,b})\le (1-1/(c_{\text{unsat}}\cdot r\log n))$ for all $|\alpha|=t$.

    \end{enumerate}
    We prove that $\widetilde{B}\cap \range(C)=\emptyset$.
    Indeed, fix any $b\in \widetilde{B}$ and observe that for each $x\in \ZO^n$,
    \begin{align*}
    \langle C(x),b\rangle&=\frac{1}{m}\sum_{i=1}^m b_ig(x_{p^i_1},\ldots,x_{p^i_{2^t-1}})\\
    &=\frac{1}{m}\sum_{i=1}^m b_i\sum_{\alpha\subseteq[2^t-1]} \widehat{g_\alpha}\prod_{j\in \alpha} x_{p^i_j}\\
    &=\frac{1}{m}\sum_{\lvert\alpha\rvert\le t} \widehat{g_\alpha} \sum_{i=1}^m b_i\prod_{j\in \alpha} x_{p^i_j}\\
    &=\sum_{\lvert\alpha\rvert < t} \widehat{g_\alpha}\cdot\psi_{\cH_{\alpha},b}(x)+\sum_{\lvert\alpha\rvert = t} \widehat{g_\alpha}\cdot\psi_{\cH_{\alpha},b}(x)\\
    &\le \lVert \widehat{g} \rVert_0\cdot \frac{1}{2^{2t+1}\cdot c_{\text{unsat}}\cdot r\log n} + \sum_{\lvert\alpha\rvert=t} \lvert \widehat{g_\alpha}\rvert\cdot (1-1/(c\cdot r\log n)) & (b\in \widetilde{B}) \\
    &\le (1/(2c_{\text{unsat}}\cdot r\log n))+(1-1/(c_{\text{unsat}}\cdot r\log n)) & (\text{\Cref{lemma: DT_Fourier_weight}})\\
    &<1.
    \end{align*}

    Lastly, by \Cref{thm: refute semi random xor almost independence}, we have $\Pr_{\bm{b}\sim \cD}[\val(\psi_{\cH_{\alpha},\bm{b}})\ge \frac{1}{2^{2t+1}\cdot c_{\text{unsat}}\cdot r\log n}]\le \frac{1}{\poly(n)}$ for every $\lvert\alpha\rvert<t$ provided that
    \begin{align*}
    m&\ge c_{\text{adaptive}}\cdot n^{t/2-\frac{t-2}{2(t+2+4\cdot\mathbf{1}\{t \text{ is even}\})}}\cdot \log^5 n \\
    &\ge c_{\text{refute}}\cdot n^{(t-1)/2}\cdot \log n\cdot (2c_{\text{unsat}}\cdot r\log n)^{2+2\cdot \mathbf{1}\{t \text{ is even}\}},
    \end{align*}
    and by ~\Cref{thm: weak_refutation}, we have $\Pr_{\bm{b}\sim \cD}[
    \val(\psi_{\cH_\alpha,\bm{b}})\ge (1-1/(c_{\text{unsat}}\cdot r\log n))]\le \frac{1}{\poly(n)}$ for every $\lvert \alpha \rvert=t$ provided that
    \[
    m\ge c_{\text{adaptive}}\cdot n^{t/2-\frac{t-2}{2(t+2+4\cdot\mathbf{1}\{t \text{ is even}\})}}\cdot \log^5 n     \ge c_{\text{weak}}\cdot n(n/r)^{k/2-1}\log n.
    \]
    Together with the union bound, we establish that
    \[
    \Pr_{\bm{b}\sim \cD}[\bm{b}\in \range(C)]\le 1-\Pr_{\bm{b}\sim \cD}[\bm{b}\in \widetilde{B}]\le 2^{2t}\cdot \frac{1}{\poly(n)}+2^{2t}\cdot \frac{1}{\poly(n)}\le \frac{1}{\poly(n)}. \qedhere
    \]
    \end{proof}

\section{Algorithms for Range Avoidance}

\subsection{Polynomial Time Algorithm for Stretch \texorpdfstring{$m=\tilde{\Omega}(n^{(t-1)/2})$}{~}}
\label{sec: poly_algo_for_avoid}
\PolyAlgoForAvoid*
We will only briefly sketch the algorithm here as most of the ideas have been illustrated in~\cite{KPI25} and the proof of~\Cref{thm:nonadaptive-bit-probe}.

As pointed out by~\cite[Lemma 7]{KPI25},~\Cref{thm:nonadaptive-bit-probe} immediately yields an $\FP^{\NP}$ algorithm for $\NC^0_t$-Avoid with stretch $m=\tilde{\Omega}(n^{(t-1)/2})$.
In slightly more detail, let $G:\{0,1\}^s\to \{0,1\}^m$ be an explicit generator for a $(c_{\text{biased}}\cdot n)^{-t}$-biased distribution with seed length $s=O(\log m+t\log n)$, where $c_{\text{biased}}$ is a constant as in~\Cref{thm:nonadaptive-bit-probe}.
Observe that $|\range(G)|\le 2^s=\poly(m,n^t)$.
\Cref{thm:nonadaptive-bit-probe} guarantees that there always exists some $b\in \range(G)$ such that $b\notin \range(C)$.
Hence, one can go over every $b\in \range(G)$, and use the $\NP$ oracle to check if $b$ is present in the range of $C$.

Korten, Pitassi, and Impagliazzo~\cite{KPI25} further get rid of the $\NP$ oracle used in their algorithm for stretch $m=\tilde{\Omega}(n^{t/2})$ by observing the following:
When proving the corresponding data structure lower bound, they certify if $b\notin \range(C)$ by upper bounding the maximum correlation by a nonnegative combination of refutation bounds of weighted XORs. As these latter bounds are computed in polynomial time, so is the certification.

We can remove the $\NP$ oracle in our algorithm in the same manner. Following are more details.



If there are at least $n+1$ output bits are XOR or NXOR, we can simply find one solution outside the range by performing Gaussian elimination.

Otherwise, only an $o(1)$-fraction of output bits are XOR or NXOR, so we may assume that none of the output bits are.
We construct the hypergraph $\cH_\alpha$ and the vector $\widehat{g_\alpha}$ for every $\alpha\subsetneq [t]$ as in the proof of~\Cref{thm:nonadaptive-bit-probe}.
Afterwards, we enumerate each $b\in \range(G)$.
For every $\alpha\subsetneq [t]$, we run the refutation algorithm on $(\cH_\alpha, b\odot \widehat{g_\alpha})$ given by~\Cref{thm: refute semi random xor almost independence} and obtain $\REF(\psi_{\cH_{\alpha},b\odot \widehat{g_\alpha}})\ge\val(\psi_{\cH_{\alpha},b\odot \widehat{g_\alpha}})$. 
Then we inspect if
\begin{equation}\label{eqn:sum_of_REF_less_than_m}
\sum_{\alpha\subsetneq [t]} |\REF(\psi_{\cH_\alpha,b\odot \widehat{g_\alpha}})|+(1-2^{1-t}) <1.
\end{equation}
If this is indeed the case, we can then claim that $b\notin \range(C)$ since the left-hand side of \eqref{eqn:sum_of_REF_less_than_m} is at least the maximum correlation $\max_{x\in \{\pm 1\}^n} \langle C(x),b\rangle$.
Lastly, as demonstrated in the proof of~\Cref{thm:nonadaptive-bit-probe},~\eqref{eqn:sum_of_REF_less_than_m} holds for constant fraction of $b\in \range(G)$, hence,  
the algorithm can always find a solution.

\subsection{Subexponential Time Algorithm for Arbitrary Stretch}

\SubEXPAlgoForAvoid*
The algorithm is nearly identical to the one presented in~\cref{sec: poly_algo_for_avoid}, with only two differences:
\begin{itemize}
    \item Let $r\coloneqq n^{1-\frac{2\epsilon}{t-3}+o(1)}$. We enumerate $b$ from the range of an explicit generator $G:\{0,1\}^s\to \{0,1\}^m$ for a $2\lceil r \log n\rceil$-wise independent distribution with seed length $s=O(r\log n\cdot \log m)$, instead of a small biased one.
    \item For each weighted $|\alpha|$-XOR instance $\psi_{\cH_\alpha,b\odot \widehat{g_\alpha}}$, we run the refutation algorithm with Kikuchi level $r$. 
\end{itemize}

By~\Cref{thm: refute semi random xor}, for each $\alpha\subsetneq [t]$, we have $\REF(\psi_{\cH_\alpha,\bm{b}\odot \widehat{g_\alpha}})\le 2^{-2t}$ with probability $1-\frac{1}{\poly(n)}$ provided that $m\ge n^{1+\epsilon}\ge n\cdot (n/r)^{(t-1)/2-1}\cdot \log n\cdot (2^{2t})^4$.
Following the same line of reasoning, we can show that the algorithm always succeeds in finding a solution.

The running time of the algorithm is $2^{s+\tilde{O}(r)}=2^{n^{1-\frac{2\epsilon}{t-3}+o(1)}}$.

\bibliographystyle{alpha}
\bibliography{references}

\end{document}